\documentclass[aps,pra,reprint,nofootinbib,superscriptaddress,floatfix,flushbottom]{revtex4-2}

\usepackage{amsmath,amssymb,amsthm,mathtools,bm}
\usepackage{mathrsfs}
\usepackage{booktabs}
\usepackage{microtype}
\usepackage{tikz}
\usetikzlibrary{arrows.meta,positioning,fit,decorations.pathreplacing}
\usepackage[colorlinks=true,linkcolor=blue,citecolor=blue,urlcolor=blue]{hyperref}

\newtheorem{theorem}{Theorem}
\newtheorem{lemma}{Lemma}
\newtheorem{corollary}{Corollary}
\newtheorem{proposition}{Proposition}

\theoremstyle{definition}
\newtheorem{example}{Example}
\theoremstyle{plain}

\newcommand{\Tr}{\operatorname{Tr}}
\newcommand{\id}{\operatorname{id}}
\newcommand{\Ad}{\operatorname{Ad}}
\newcommand{\MRM}{\operatorname{MRM}}
\newcommand{\Sp}{\operatorname{Sp}}
\newcommand{\CSp}{\operatorname{CSp}}
\newcommand{\Cent}{\operatorname{Cent}}
\newcommand{\cM}{\mathcal M}
\newcommand{\cN}{\mathcal N}
\newcommand{\cT}{\mathcal T}
\newcommand{\cE}{\mathcal E}
\newcommand{\cD}{\mathcal D}
\newcommand{\F}{\mathbb F_d}
\newcommand{\ketbra}[2]{|#1\rangle\!\langle #2|}
\newcommand{\doteqph}{\mathrel{\dot=}}
\newcommand{\ket}[1]{|#1\rangle}
\newcommand{\bra}[1]{\langle #1|}

\hypersetup{pdftitle={Transport and Channel Normal Forms for Clifford Interactions},pdfauthor={Chunhe Xiong, Sunho Kim, Qing-Hua Zhang, Shao-Ming Fei, Junde Wu}}

\begin{document}

\title{Transport and Channel Normal Forms for Clifford Interactions}

\author{Chunhe Xiong}
\email{xiongchunhe@csu.edu.cn}
\affiliation{School of Mathematics and Statistics, Central South University, Changsha 410083, China}

\author{Sunho Kim}
\email{kimsunho81@hrbeu.edu.cn}
\affiliation{School of Mathematical Sciences, Harbin Engineering University, Harbin 150001, China}

\author{Qing-Hua Zhang}
\affiliation{School of Mathematics and Statistics, Changsha University of Science and Technology, Changsha 410114, China}

\author{Shao-Ming Fei}
\email{feishm@cnu.edu.cn}
\affiliation{School of Mathematical Sciences, Capital Normal University, Beijing 100048, China}

\author{Junde Wu}
\email{wjd@zju.edu.cn}
\affiliation{School of Mathematical Sciences, Zhejiang University, Hangzhou 310027, China}


\begin{abstract}
We study how Clifford interactions with a fixed environment determine quantum transmission, in a setting that extends positive discrete quantum convolution.
We characterize when Clifford operations on the environment can be replaced by operations on the signal input and the two outputs, and classify the interactions under local Clifford transformations. For pure stabilizer environments, we determine the noiseless, dephasing, and completely depolarizing parts of the channel directly from the interaction and environmental stabilizers, obtaining its quantum and private capacities. This reveals a distinction between interactions: those locally equivalent to positive convolution produce entanglement-breaking channels for every pure stabilizer environment, whereas every other interaction in the class admits noiseless quantum transmission with a suitable pure stabilizer environment. For positive convolution, environmental Clifford transport also allows us to separate the stabilized coordinates of an arbitrary mixed environment. These coordinates carry only a classical label shared by the receiver and environment, so removing them preserves quantum and private capacities. The resulting channel reduction yields a private-capacity bound by the modified relative entropy of magic of the environmental state.

\end{abstract}

\maketitle

\section{Introduction}

An interaction with an environment can preserve quantum information,
retain only classical information, or erase the input.  Which of these
mechanisms occurs depends on how the coupling distributes input
observables between the outputs, as well as on the environmental state.
A channel normal form separates these mechanisms and identifies the
degrees of freedom available for communication.  We seek such a
description for channels induced by bipartite Clifford interactions.

Discrete quantum convolution provides a setting for this problem.
Bu, Gu, and Jaffe introduced finite-field analogues of beam-splitter
and amplifier interactions and developed their stabilizer mean states,
entropy inequalities, and central-limit theory
\cite{BuGuJaffePNAS2023,BuGuJaffeDVG2023,BuGuJaffeTesting2025}.
For positive convolution couplings, pure stabilizer environments produce
entanglement-breaking channels, whereas nonstabilizer environments can
support quantum communication
\cite{BuJaffePRL2025,SunJinPRA2025,XiongKimLongWu2026}.
This environmental dependence parallels the use of non-Gaussian
environments in bosonic communication
\cite{LamiPlenioGiovannettiHolevo2020}.
The coupling also matters: within the componentwise convolution family,
different vanishing-entry patterns select entanglement-breaking or
generalized-dephasing regimes~\cite{XiongCoupling2026}.

In the BGJ many-qudit construction, the same coupling acts on every
corresponding pair of signal and environmental qudits.  We allow different
pairwise couplings and collective Clifford interactions, while retaining
linear mixing of Weyl observables in phase space.  This makes it possible
to distinguish properties of the common coupling from those of Clifford
mixing itself.  The two registers have equal size and odd-prime local
dimension.  We require each nonzero phase-space direction at either
input to contribute to both outputs.  This regularity condition includes
positive BGJ convolution and will be stated as invertibility of the four blocks of the symplectic matrix.  Clifford interactions retain an exact multiplication
rule for the input characteristic functions, so their transmitted
observables can be analyzed directly through finite-dimensional
symplectic geometry.

The distinction is already visible with two qudits in each register.
For one fixed interaction, a product pure stabilizer environment gives
zero quantum capacity, whereas a Bell stabilizer environment permits
noiseless transmission of one qudit
(Example~\ref{ex:two-sector-code}).  Both environments have zero
modified relative entropy of magic and are related by an environmental
Clifford.  Thus an environmental Clifford need not preserve communication
at a fixed interaction.  The correlations here are within the environment;
the signal and environment remain initially independent.

We first determine when an environmental Clifford can be compensated
by operations on the signal input and the two outputs.  This
\emph{transport} condition and the classification of interactions under
separate Clifford operations at their ports are governed by a matrix
invariant, the cross-ratio.  For a specified pure stabilizer environment,
the same invariant gives the numbers of noiseless, dephasing, and
completely depolarizing factors in the induced channel.  Our rank formula
computes these numbers from the interaction and environmental stabilizers;
the proof constructs the corresponding encoding and decoding and yields
the quantum and private capacities.  The abstract stabilizer-channel
normal form was established by Looi and Griffiths~\cite{LooiGriffiths2011}.
Here its factors are determined directly from the coupling data.

Positive BGJ convolution occupies a distinguished part of this
classification.  Its local Clifford orbits are precisely those for
which every environmental Clifford is transportable, and also those
for which every pure stabilizer environment gives an entanglement-breaking
channel.  In this class, transport separates the stabilized coordinates
of an arbitrary environment, including a mixed state.  These coordinates
carry a classical label shared by receiver and environment; removing
them preserves quantum and private capacities.  The remaining channel
gives a private-capacity bound by the environmental modified relative
entropy of magic.  Outside this class, a suitable pure stabilizer
environment already supports noiseless quantum transmission.

The geometric ingredients are operator cross-ratios and symplectic
self-adjoint operators~\cite{Zelikin2006,GoldsteinGuralnick2007}.
Related commutation-form methods describe party-local Clifford
equivalence of stabilizer states~\cite{EnglbrechtKraftKraus2022},
and Sun, Jin, and Jin established fixed-environment channel Clifford
equivalences within their stabilizer-convolution family
\cite{SunJinJin2026}.
Our four-port classification includes the environmental coordinate
change and allows nonscalar cross-ratios.
Bu and Jaffe established an environmental MRM upper bound
on the quantum capacity of discrete beam-splitter
channels~\cite{BuJaffePRL2025}.
Our channel reduction applies to every positive convolution
coupling and to arbitrary pure or mixed environmental states,
and identifies the optimized private and coherent information
with those of a residual channel at every block length.
It extends the quantum-capacity bound to this coupling class
and yields a private-capacity bound with the same coefficient.
In the pure-environment beam-splitter setting, the latter
improves the coefficient in the bound of
Xiong et al.~\cite{XiongKimLongWu2026} from two to one.


Section~\ref{sec:preliminaries} defines the interactions and basic
identities.  Section~\ref{sec:general-transport} gives the transport
criterion, interaction classification, and pure-stabilizer channel
normal form.  Section~\ref{sec:factorization} treats environmental
reduction and capacity bounds in the positive-convolution class.
Longer proofs are collected in the appendices.

\section{Clifford interactions and convolution channels}
\label{sec:preliminaries}

Let $d$ be an odd prime, let $\F$ be the field with $d$ elements,
and let $\mathcal H_n=(\mathbb C^d)^{\otimes n}$.
All phase-space ranks and dimensions are over $\F$.
Logarithms are to base two.  We use von Neumann entropy and quantum
relative entropy in the conventions
\begin{align*}
 S(\rho)&=-\Tr(\rho\log_2\rho),\\
 D(\rho\Vert\omega)&=\Tr[\rho(\log_2\rho-\log_2\omega)],
\end{align*}
where the latter is infinite unless
$\operatorname{supp}\rho\subseteq\operatorname{supp}\omega$.

\subsection{Phase space and Clifford unitaries}
\label{sec:weyl-clifford}

The phase space is $V=V_n=\F^n\oplus\F^n$.  Its vectors are
$z=(p,q)$ and $z'=(p',q')$, written as columns, with
$p,q,p',q'\in\F^n$.  The standard symplectic form is
\[
 [z,z']=z^{\mathsf T}Jz'=p\cdot q'-q\cdot p',
 ~ J=\begin{pmatrix}0&I_n\\-I_n&0\end{pmatrix}.
\]
This bilinear form is alternating and nondegenerate.

For a subspace $W\subset V$, write
\[
 W^\perp=\{z\in V:[z,w]=0\text{ for every }w\in W\}.
\]
The restricted form on $W$ is nondegenerate when $W\cap W^\perp=\{0\}$.
Let $\zeta=e^{2\pi i/d}$.  The generalized Pauli phase and shift
matrices on one qudit are
\[
 Z|j\rangle=\zeta^j|j\rangle,~
 X|j\rangle=|j+1\rangle,~ j\in\F.
\]
For $p,q\in\F^n$, set
$Z(p)=\bigotimes_{j=1}^n Z^{p_j}$ and
$X(q)=\bigotimes_{j=1}^n X^{q_j}$.
The Weyl operator with phase-space label $z=(p,q)$ is
\begin{equation}
 w(z)=w(p,q)=\zeta^{-2^{-1}p\cdot q}Z(p)X(q),
 \label{eq:balanced-weyl}
\end{equation}
where $2^{-1}$ is the inverse of $2$ in $\F$.  Thus $p$ labels phase
changes and $q$ labels computational-basis shifts.  The Weyl operators satisfy
\begin{align*}
& w(z)w(z')=\zeta^{2^{-1}[z,z']}w(z+z'),\\
& w(z)^\dagger=w(-z),~
 \Tr[w(z)^\dagger w(z')]=d^n\delta_{z,z'}.
\end{align*}
They form an orthogonal operator basis.  The characteristic function
and its inversion formula are
\[
 \Xi_\rho(z)=\Tr[\rho w(-z)],~
 \rho=d^{-n}\sum_{z\in V}\Xi_\rho(z)w(z).
\]
For a linear map $M$ on $V$, its symplectic adjoint is
\[
 M^\sharp=J^{-1}M^{\mathsf T}J,
 ~ [Mx,y]=[x,M^\sharp y].
\]
The symplectic group is
\[
 \Sp(V)=\{F\in\operatorname{GL}(V):F^{\mathsf T}JF=J\};
\]
we also denote it by $\Sp(2n,d)$.

A Clifford unitary normalizes the Weyl operators.  Its conjugation
action has the form $Kw(z)K^\dagger\doteqph w(Fz)$ for some
$F\in\Sp(V)$, called its homogeneous action.  Here $\doteqph$
denotes equality up to a scalar phase.  A Clifford with homogeneous
action $F$ is a lift of $F$.
For each $F$ we choose an \emph{exact lift} $U_F$ satisfying
\begin{equation}
 U_Fw(z)U_F^\dagger=w(Fz),~\forall z\in V.
 \label{eq:exact-lift}
\end{equation}
To see its existence, extend $w(z)\mapsto w(Fz)$ linearly to the
Weyl basis.  The Weyl multiplication law and $[Fz,Fz']=[z,z']$
show that this map preserves products and adjoints; its inverse is
obtained from $F^{-1}$.  It is therefore a $*$-automorphism of the
full matrix algebra and is implemented by a unitary.  Two exact
lifts differ only by a scalar phase, since their quotient commutes
with the whole Weyl basis.  The same construction applies to the
joint phase space $V\oplus V$.
We write $\Ad_U(X)=UXU^\dagger$.

\emph{Symplectic sectors and quantum subsystems.}
Write $V=V_1\oplus^\perp V_2$ for a direct sum of symplectically
orthogonal subspaces, with $\dim V_i=2n_i$.
Each $V_i$ is nondegenerate: a vector in $V_i$ orthogonal to $V_i$
is also orthogonal to the other sector, hence to all of $V$.
Choose symplectic identifications $T_i:V_{n_i}\to V_i$ and put
$T(z_1,z_2)=T_1z_1+T_2z_2$.
Its Clifford realization obeys
\begin{equation}
 U_T\bigl(w_1(z_1)\otimes w_2(z_2)\bigr)U_T^\dagger
 =w\bigl(T(z_1,z_2)\bigr).
 \label{eq:sector-weyl-factorization}
\end{equation}
Thus the Weyl operators in the two sectors act on separate tensor
factors $\mathcal H_{n_1}\otimes\mathcal H_{n_2}$ in these coordinates.
For a state $\rho$, its sector reductions are the partial traces of
$U_T^\dagger\rho U_T$.  These factors may be logical subsystems,
rather than the original physical grouping of qudits.  Changing the
adapted basis within either sector acts by a local Clifford, so the
sector entropies do not depend on that choice.

\subsection{Interactions, channels, and basic identities}
\label{sec:interaction-definitions}

The signal input $A$ and environment input $B$ each contain $n$
qudits, as do the outputs $A',B'$.  In port order the joint symplectic
form is $J\oplus J$.  Let
\begin{gather*}
 S=\begin{pmatrix}A&B\\C&D\end{pmatrix}\in\Sp(V\oplus V),\\
 w_{AB}(z_A,z_B)=w(z_A)\otimes w(z_B),
\end{gather*}
and choose its exact lift $U_S$ as in Eq.~\eqref{eq:exact-lift}.
Other lifts with the same homogeneous action differ by local output
Weyl operators and a scalar phase, as shown in
Appendix~\ref{app:clifford-lifts}; their channel capacities are unchanged.
The columns of $S$ refer to the input ports and the rows to the
output ports.  We call $S$ \emph{regular} if all four blocks are
invertible, and write $\Sp_{\rm reg}(V\oplus V)$ for this set.
Equivalently, every nonzero Weyl direction supported on a single
input port has a nonzero component at each output.  The classification
below concerns regular interactions; the first two identities in
Proposition~\ref{prop:basic-identities} hold without this restriction.

For independent input states $\rho_A,\sigma_B$, define
\begin{equation}
 \rho\boxtimes_S\sigma
 :=\cN_{S,\sigma}(\rho)
 =\Tr_{B'}[U_S(\rho\otimes\sigma)U_S^\dagger].
 \label{eq:general-convolution}
\end{equation}
Figure~\ref{fig:stinespring} shows this channel as a two-input
interaction followed by discarding the environmental output.
For a purification $|\phi_\sigma\rangle_{BR}$, the Stinespring
isometry~\cite{Stinespring1955,Watrous2018} is
\[
 V_{S,\sigma}|\psi\rangle_A
 =(U_S\otimes I_R)(|\psi\rangle_A\otimes|\phi_\sigma\rangle_{BR}).
\]
The receiver and complementary channels are
\begin{align}
 \cN_{S,\sigma}(\rho)
 &=\Tr_{B'R}[V_{S,\sigma}\rho V_{S,\sigma}^\dagger],\nonumber\\
 \widehat\cN_{S,\sigma}(\rho)
 &=\Tr_{A'}[V_{S,\sigma}\rho V_{S,\sigma}^\dagger].
 \label{eq:true-complement}
\end{align}
The complementary output is $B'R$, including the purification
reference when $\sigma$ is mixed.
Fixing the pure ancillary state $|\phi_\sigma\rangle_{BR}$ turns
the two-input unitary $U_S$ into the isometry $V_{S,\sigma}$ from
$A$ to $A'B'R$.  Each channel use receives an independent copy
of $\sigma$; inputs may be entangled across uses.

\begin{figure}[tb]
\centering
\begin{tikzpicture}[>=Latex,every node/.style={font=\small},
                    line width=0.5pt]
 \node at (0,1) {$\rho_A$};
 \node at (0,0) {$\sigma_B$};
 \node[draw,minimum width=1.1cm,minimum height=1.6cm]
       (interaction) at (2,0.5) {$U_S$};
 \draw[->] (0.45,1)--node[above] {$A$}(1.45,1);
 \draw[->] (0.45,0)--node[above] {$B$}(1.45,0);
 \draw[->] (2.55,1)--node[above] {$A'$}(4.7,1)
       node[right] {$\cN_{S,\sigma}(\rho)$};
 \draw[->] (2.55,0)--node[above] {$B'$}(3.6,0);
 \draw (3.6,-0.3) rectangle (4.3,0.3);
 \draw (3.72,-0.12) arc[start angle=180,end angle=0,radius=0.23];
 \draw (3.95,-0.12)--(4.1,0.13);
 \draw[->,dashed] (4.3,0)--(4.95,0);
 \node[right,font=\footnotesize] at (4.95,0) {discard};
 \node at (2.8,-0.9) {$\equiv$};
 \node at (0,-1.8) {$\rho_A$};
 \node[draw,minimum width=1.5cm,minimum height=0.65cm]
       (channel) at (2,-1.8) {$\cN_{S,\sigma}$};
 \draw[->] (0.45,-1.8)--(channel.west);
 \draw[->] (channel.east)--(4.7,-1.8)
       node[right] {$\cN_{S,\sigma}(\rho)$};
\end{tikzpicture}
\caption{Two-input realization of $\cN_{S,\sigma}$.  Fixing the
environment $\sigma_B$ and discarding $B'$ gives the channel shown
below.  The meter denotes a complete projective measurement on
$B'$ with both the measured system and its outcome discarded,
equivalently the partial trace $\Tr_{B'}$.  No outcome is selected.}
\label{fig:stinespring}
\end{figure}
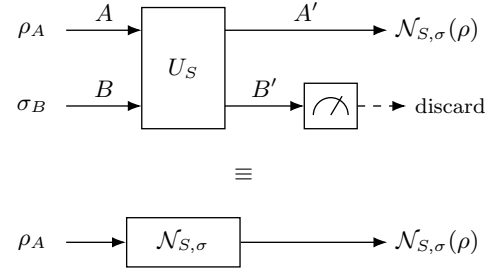

\emph{BGJ convolution as a special case.}
For
\[
 G=\begin{pmatrix}g_{00}&g_{01}\\g_{10}&g_{11}\end{pmatrix}
 \in\operatorname{GL}(2,\F),~ N=(\det G)^{-1},
\]
the BGJ interaction is the computational-basis permutation
\begin{equation}
 U_G|i,j\rangle
 =|Ng_{11}i-Ng_{10}j,\,-Ng_{01}i+Ng_{00}j\rangle,
 \label{eq:UG}
\end{equation}
with $i,j\in\F^n$ and the same $G$ acting on every coordinate pair.
The basis labels transform by $G^{-\mathsf T}\otimes I_n$, so
the grouped Weyl labels $p=(p_A,p_B)$ and $q=(q_A,q_B)$ transform by
$G\otimes I_n$ and $G^{-\mathsf T}\otimes I_n$, respectively.
In port order this gives $S_G$ with blocks
\begin{align}
 A_G&=\operatorname{diag}(g_{00}I_n,Ng_{11}I_n),\nonumber\\
 B_G&=\operatorname{diag}(g_{01}I_n,-Ng_{10}I_n),\nonumber\\
 C_G&=\operatorname{diag}(g_{10}I_n,-Ng_{01}I_n),\nonumber\\
 D_G&=\operatorname{diag}(g_{11}I_n,Ng_{00}I_n).
 \label{eq:convolution-blocks}
\end{align}
These transformations preserve $p\cdot q$ and hence the Weyl phase
in Eq.~\eqref{eq:balanced-weyl}.  Thus $U_G$ is an exact lift of
$S_G$, and we choose $U_{S_G}=U_G$.  We retain the BGJ channel notation
\[
 \Lambda_{G,\sigma}:=\cN_{S_G,\sigma},~
 \rho\boxtimes_G\sigma=\rho\boxtimes_{S_G}\sigma.
\]
With the same purification, also
$V_{G,\sigma}=V_{S_G,\sigma}$ and
$\widehat\Lambda_{G,\sigma}=\widehat\cN_{S_G,\sigma}$.

Following Ref.~\cite[Definition~38]{BuGuJaffeDVG2023}, an invertible
$G$ is nontrivial if at most one entry vanishes.  A nontrivial $G$ is
\emph{odd-parity positive} if $g_{01},g_{10}\ne0$ and
\emph{even-parity positive} if $g_{00},g_{11}\ne0$.
It is \emph{positive} when both conditions hold.
Equation~\eqref{eq:convolution-blocks} shows that $S_G$ is regular
exactly when $G$ is positive.  General regular $S$ need not have
this componentwise form.

For BGJ convolution, characteristic-function multiplication and
entropy monotonicity are established in
Ref.~\cite{BuGuJaffeDVG2023}.  We record their form for general
Clifford interactions, together with Weyl covariance, for use in
the stabilizer normal form and environmental capacity bounds.

\begin{proposition}[Basic convolution identities]
\label{prop:basic-identities}
Let $S\in\Sp(V\oplus V)$ and use its exact lift in
Eq.~\eqref{eq:general-convolution}.
\begin{enumerate}
 \item For all states $\rho,\sigma$,
 \begin{equation}
  \Xi_{\rho\boxtimes_S\sigma}(z)
  =\Xi_\rho(A^\sharp z)\Xi_\sigma(B^\sharp z).
  \label{eq:characteristic-product}
 \end{equation}
 \item For all $u,v\in V$, Weyl displacements satisfy
 \begin{align}
  \bigl(\Ad_{w(u)}\rho\bigr)\boxtimes_S
       \bigl(\Ad_{w(v)}\sigma\bigr)=\Ad_{w(Au+Bv)}(\rho\boxtimes_S\sigma).
  \label{eq:weyl-covariance}
 \end{align}
 \item If $A$ is invertible, $(I/d^n)\boxtimes_S\sigma=I/d^n$.
 If $B$ is invertible, the maximally mixed environment gives the
 completely depolarizing channel:
 \begin{equation}
  \cN_{S,I/d^n}(X)=\Tr(X)I/d^n.
  \label{eq:depolarizing-example}
 \end{equation}
 If both blocks are invertible, then
 \begin{equation}
  S(\rho\boxtimes_S\sigma)
  \ge\max\{S(\rho),S(\sigma)\}.
  \label{eq:entropy-increase}
 \end{equation}
\end{enumerate}
The same statements apply to direct sums of interactions, with
arbitrary correlated states within either input port.
\end{proposition}

\begin{proof}
Symplecticity gives
\[
 S^{-1}=(J\oplus J)^{-1}S^{\mathsf T}(J\oplus J)
 =\begin{pmatrix}A^\sharp&C^\sharp\\B^\sharp&D^\sharp\end{pmatrix}.
\]
Conjugating $w(-z)\otimes I$ backwards through the exact lift gives
\begin{align*}
 \Xi_{\rho\boxtimes_S\sigma}(z)
 &=\Tr[(\rho\otimes\sigma)
       (w(-A^\sharp z)\otimes w(-B^\sharp z))]\\
 &=\Xi_\rho(A^\sharp z)\Xi_\sigma(B^\sharp z),
\end{align*}
proving the first identity.  The forward Weyl action gives
\[
 U_S(w(u)\otimes w(v))
 =\bigl(w(Au+Bv)\otimes w(Cu+Dv)\bigr)U_S.
\]
Taking the partial trace proves the second identity.  This covariance
also holds for any other Clifford lift with homogeneous action $S$,
since the extra phases cancel in conjugation.

For the third part, $\Xi_{I/d^n}(z)=\delta_{z,0}$ and the first
identity give the two maximally mixed output statements.  Linearity
extends the second one from density operators to every operator $X$.
Fixing either input defines a CPTP map $\mathcal T$ of the other.
Under the stated invertibility assumptions, these maps preserve
$I/d^n$.  Relative-entropy data processing yields
\begin{align*}
 n\log_2d-S(\mathcal T(\omega))
 &=D(\mathcal T(\omega)\Vert\mathcal T(I/d^n))\\
 &\le D(\omega\Vert I/d^n)
 =n\log_2d-S(\omega).
\end{align*}
Applying this to the two maps proves Eq.~\eqref{eq:entropy-increase}.
For several interactions, regroup their direct sum into the two
input ports.  The same proof applies to the joint input states,
without a product assumption within a port.
\end{proof}

Weyl covariance holds for every interaction here.  Transport of a
general environmental Clifford requires a further condition on $S$,
which is determined in Sec.~\ref{sec:general-transport}.

\subsection{Stabilizer states and MRM}
\label{sec:stabilizer-notation}

A subspace of $V$ is isotropic if the symplectic form vanishes on it,
and Lagrangian if it is isotropic of dimension $n$; then $L^\perp=L$.
The stabilizer subspace of a state is
\[
 S_\rho=\{z\in V:|\Xi_\rho(z)|=1\}.
\]
It is isotropic, and the phases of $\Xi_\rho$ on it form a character
specifying a joint stabilizer eigenspace
\cite{BuGuJaffeDVG2023,BuGuJaffeTesting2025}.
A minimal stabilizer-projection state (MSPS) is the normalized
projector onto a joint eigenspace of a commuting Weyl subgroup,
including $I/d^n$ for the trivial subgroup
\cite[Definition~6]{BuGuJaffeDVG2023}.  Pure stabilizer states are
the rank-one members of this set.

For a Lagrangian $L$, the pure stabilizer state with trivial character is
\[
 \psi_L^0=d^{-n}\sum_{t\in L}w(t),~
 \Xi_{\psi_L^0}(z)=\mathbf1_L(z),
\]
where $\mathbf1_L$ is the indicator of $L$.  Every pure stabilizer
state $\psi_L$ with this subspace has the form
$\psi_L=\Ad_{w(u)}(\psi_L^0)$ for some $u\in V$:
nondegeneracy makes $u\mapsto[u,\cdot]|_L$ onto $L^*$, so Weyl
displacements realize all characters of $L$.

The mean state $\cM(\rho)$ retains precisely the unit-modulus coefficients:
\begin{equation}
 \Xi_{\cM(\rho)}(z)=
 \begin{cases}
  \Xi_\rho(z),&z\in S_\rho,\\
  0,&z\notin S_\rho.
 \end{cases}
 \label{eq:mean-state}
\end{equation}
The modified relative entropy of magic (MRM) used here is the
relative-entropy distance to MSPS introduced in
Refs.~\cite{BuGuJaffePNAS2023,BuGuJaffeDVG2023}:
\[
 \MRM(\rho):=\min_{\omega\in\mathrm{MSPS}}D(\rho\Vert\omega).
\]
The mean state is an MSPS and the unique minimizer, giving
\[
 \MRM(\rho)=D(\rho\Vert\cM(\rho))
 =S(\cM(\rho))-S(\rho).
\]
These facts are proved in Ref.~\cite[Theorem~22]{BuGuJaffeDVG2023};
the support argument is also given in Appendix~\ref{app:normal-form}.
The MRM also decomposes into entropy increments along stabilizer
dephasings and a terminal classical entropy deficit
\cite{XiongKimZhangFeiWu2026}.
Every MSPS is a convex mixture of pure stabilizer states, but such a
mixture need not be an MSPS.  Thus positive MRM does not imply that
a state lies outside the convex hull of pure stabilizer states.
Its dependence on exact stabilizer constraints also makes it
discontinuous in general.  Clifford conjugation preserves MRM.
The tensor-product characteristic function gives
$\cM(\rho\otimes\omega)=\cM(\rho)\otimes\cM(\omega)$,
so MRM is additive.

\subsection{Communication quantities}

Capacities are measured in bits or qubits per channel use.
For a channel $\cN$, a complementary channel $\widehat\cN$,
and an ensemble $\mathsf E=\{p_u,\rho_u\}$, let
\[
 \chi(\mathsf E;\cN)
 =S\!\left(\sum_up_u\cN(\rho_u)\right)
  -\sum_up_uS(\cN(\rho_u)).
\]
The optimized private and coherent information are
\begin{align*}
 P^{(1)}(\cN)&=\sup_{\mathsf E}
 [\chi(\mathsf E;\cN)-\chi(\mathsf E;\widehat\cN)],\\
 Q^{(1)}(\cN)&=\sup_\rho I_c(\rho,\cN),\\
 I_c(\rho,\cN)&=S(\cN(\rho))-S(\widehat\cN(\rho)).
\end{align*}
Their regularizations give the private and quantum capacities
\cite{CaiWinterYeung2004,Devetak2005}:
\begin{equation}
 P(\cN)=\sup_{r\ge1}\frac{P^{(1)}(\cN^{\otimes r})}{r},
 ~
 Q(\cN)=\sup_{r\ge1}\frac{Q^{(1)}(\cN^{\otimes r})}{r}.
 \label{eq:PQ-capacity}
\end{equation}
For comparison, write $\chi^*(\cN)=\sup_{\mathsf E}\chi(\mathsf E;\cN)$
and $C(\cN)=\sup_{r\ge1}\chi^*(\cN^{\otimes r})/r$ for the
classical capacity~\cite{Holevo1998,SchumacherWestmoreland1997}.
We use the standard inequalities $Q\le P\le C$.
A channel is entanglement breaking (EB) if its output is separable
from any reference for every joint input state
\cite{HorodeckiShorRuskai2003}.

\section{Clifford transport and port-local equivalence}
\label{sec:general-transport}
\label{sec:transport}

Throughout this section, $S$ is an arbitrary regular bipartite
Clifford interaction as defined in Sec.~\ref{sec:interaction-definitions}.
No componentwise BGJ form is assumed.  We first determine which
environmental changes can be compensated at the other ports, then
classify the interactions and compute their channels for pure
stabilizer environments.

An environmental Clifford is \emph{transportable} if its insertion
at $B$ can be replaced by a Clifford at $A$ and separate Cliffords
at $A',B'$, with the interaction fixed.  Here \emph{port-local}
permits arbitrary Clifford operations within each port, including
operations entangling its qudits.

\subsection{The transport criterion}

Define the cross-ratio and its symplectic centralizer by
\begin{align*}
 \mathscr R(S)&=B^{-1}AC^{-1}D,\\
 \Cent_{\Sp(V)}(R)&=\{F\in\Sp(V):FR=RF\}.
\end{align*}

\begin{theorem}[Centralizer criterion for homogeneous symplectic transport]
\label{thm:centralizer-transport}
Let \(S\) be regular and let \(F\in\Sp(V)\).  There exist
\(X_F,Y_F,Z_F\in\Sp(V)\) satisfying
\begin{equation}
 S(I\oplus F)=(X_F\oplus Y_F)S(Z_F\oplus I)
 \label{eq:general-transport}
\end{equation}
if and only if
\begin{equation*}
 F\mathscr R(S)=\mathscr R(S)F.
\end{equation*}
When they exist, the compensating transformations are unique:
\begin{align}
 X_F&=BFB^{-1},~
 Y_F=DFD^{-1},
 \nonumber\\
 Z_F&=C^{-1}DF^{-1}D^{-1}C.
 \label{eq:transport-compensators}
\end{align}
Consequently, the homogeneous symplectic parts that can be transported form
\(\Cent_{\Sp(V)}(\mathscr R(S))\).
\end{theorem}

The proof is in Appendix~\ref{app:centralizer-proof}.
We next lift the matrix identity to the Stinespring isometry defined in Sec. II B.  This keeps the receiver and the
full complementary output $B'R$ under the same input transformation.

\begin{corollary}[Transport of the receiver--complementary pair]
\label{cor:general-pair}
Let \(S\) be regular.  An environmental Clifford \(K\), with homogeneous
part \(F\), has a factorized transport through \(U_S\) if and only if
\(F\in\Cent_{\Sp(V)}(\mathscr R(S))\).
In this case there are Cliffords \(K_A,K_{A'},K_{B'}\) with
\begin{equation}
 U_S(I\otimes K)\doteqph
 (K_{A'}\otimes K_{B'})U_S(K_A\otimes I).
 \label{eq:general-unitary-transport}
\end{equation}
For \(\sigma'=K\sigma K^\dagger\), the corresponding isometries satisfy
\begin{equation*}
 V_{S,\sigma'}\doteqph
 (K_{A'}\otimes K_{B'}\otimes I_R)V_{S,\sigma}K_A,
\end{equation*}
when \(|\phi_{\sigma'}\rangle=(K\otimes I_R)|\phi_\sigma\rangle\).
\end{corollary}

\begin{proof}
Necessity follows from Theorem~\ref{thm:centralizer-transport}
by taking homogeneous actions.  For sufficiency, write
$K\doteqph w(v)U_F$.  Exact lifts implement the transport identity
for $F$, while Weyl covariance sends the displacement to
$w(Bv)\otimes w(Dv)$ at the outputs.  Absorbing these two Weyl
operators into the output Cliffords gives
Eq.~\eqref{eq:general-unitary-transport}.
Appendix~\ref{app:clifford-lifts} proves the decomposition of $K$
and this construction.  Applying the identity to an input vector
and the chosen purification gives the isometry relation.  Taking the
two partial traces yields
\begin{align*}
 \cN_{S,\sigma'}
 &=\Ad_{K_{A'}}\circ\cN_{S,\sigma}\circ\Ad_{K_A},\\
 \widehat\cN_{S,\sigma'}
 &=\Ad_{K_{B'}\otimes I_R}\circ
 \widehat\cN_{S,\sigma}\circ\Ad_{K_A}.
\end{align*}
\end{proof}

The input and receiver unitaries can be absorbed into encoding and
decoding, while complementary output unitaries preserve privacy leakage.
Thus transport preserves the finite-block optimizations and capacities.
Figure~\ref{fig:transport} shows the identity.

\begin{figure*}[t]
\centering
\begin{tikzpicture}[
  >=Latex,
  every node/.style={font=\small},
  gate/.style={draw,minimum width=0.85cm,minimum height=0.48cm,
               fill=blue!5,inner sep=3pt},
  interaction/.style={draw,minimum width=1cm,minimum height=1.25cm,
                     fill=blue!5,inner sep=3pt}]
 \node (lA) at (0,0) {$A$};
 \node (lB) at (0,-0.7) {$B$};
 \node (lR) at (0,-1.45) {$R$};
 \node[gate] (lK) at (1.3,-0.7) {$K$};
 \node[interaction] (lU) at (3,-0.35) {$U_S$};
 \node (lAp) at (5.1,0) {$A'$};
 \node (lBp) at (5.1,-0.7) {$B'$};
 \node (lRp) at (5.1,-1.45) {$R$};
 \draw[->] (lA)--(2.5,0);
 \draw[->] (lB)--(lK);
 \draw[->] (lK)--(2.5,-0.7);
 \draw[->] (3.5,0)--(lAp);
 \draw[->] (3.5,-0.7)--(lBp);
 \draw[->] (lR)--node[above,font=\scriptsize]{$I_R$}(lRp);
 \node[font=\large] at (6.1,-0.35) {$\doteqph$};
 \node (rA) at (7.1,0) {$A$};
 \node (rB) at (7.1,-0.7) {$B$};
 \node (rR) at (7.1,-1.45) {$R$};
 \node[gate] (rKA) at (8.3,0) {$K_A$};
 \node[interaction] (rU) at (10,-0.35) {$U_S$};
 \node[gate] (rKAp) at (11.65,0) {$K_{A'}$};
 \node[gate] (rKBp) at (11.65,-0.7) {$K_{B'}$};
 \node (rAp) at (13.2,0) {$A'$};
 \node (rBp) at (13.2,-0.7) {$B'$};
 \node (rRp) at (13.2,-1.45) {$R$};
 \draw[->] (rA)--(rKA);
 \draw[->] (rKA)--(9.5,0);
 \draw[->] (rB)--(9.5,-0.7);
 \draw[->] (10.5,0)--(rKAp);
 \draw[->] (10.5,-0.7)--(rKBp);
 \draw[->] (rKAp)--(rAp);
 \draw[->] (rKBp)--(rBp);
 \draw[->] (rR)--node[above,font=\scriptsize]{$I_R$}(rRp);
\end{tikzpicture}
\caption{Clifford transport at the four ports.  When the homogeneous action
of $K$ commutes with $\mathscr R(S)$, the two circuits agree up to an
overall phase.  The environmental operation is replaced by one signal-input
and two output Cliffords.  The purification reference $R$ is unchanged.
The operation on $B'$ expresses unitary equivalence of the complementary
output $B'R$ and preserves its information leakage.}
\label{fig:transport}
\end{figure*}
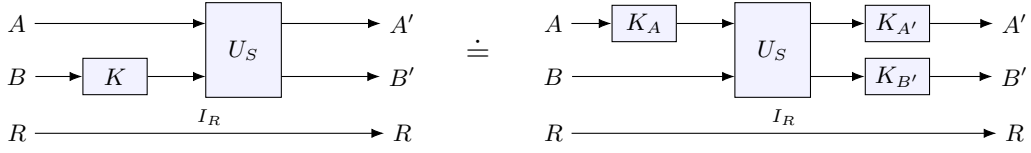

\subsection{A complete invariant for regular interactions}

The transport criterion concerns environmental changes at a fixed
interaction.  We now compare interactions under independent Clifford
operations at all four ports.  Their homogeneous actions define the group action
\begin{equation*}
 (P,Q;U,T)\boldsymbol{\cdot}S
 =(P\oplus Q)S(U\oplus T)^{-1}.
\end{equation*}
The corresponding cross-ratio changes by conjugation with \(T\).
Its possible values are constrained by the symplectic block identities:
they are self-adjoint with respect to the symplectic form, and neither
zero nor one is an eigenvalue.  We therefore write
\begin{equation*}
 \mathfrak X(V)=\{R\in\operatorname{GL}(V):
 I-R\in\operatorname{GL}(V),\ R^{\mathsf T}J=JR\}.
\end{equation*}

\begin{theorem}[Cross-ratio completeness for regular local orbits]
\label{thm:local-orbit-classification}
Two regular actions \(S_1,S_2\) are related by
\[
 S_2=(P\oplus Q)S_1(U\oplus T)^{-1},
 ~ P,Q,U,T\in\Sp(V),
\]
if and only if their cross-ratios are symplectically conjugate:
\[
 \mathscr R(S_2)=T\mathscr R(S_1)T^{-1}
 \quad\text{for some }T\in\Sp(V).
\]
Every \(R\in\mathfrak X(V)\) is realized by a regular interaction.
Consequently, \(S\mapsto\mathscr R(S)\) induces the bijection
\[
 \Sp_{\rm reg}(V\oplus V)/\Sp(V)^4
 \ \longrightarrow\
 \mathfrak X(V)/\Sp(V),
\]
where the quotient on the right is by symplectic conjugation.

\end{theorem}

Exact Clifford lifts convert this classification to ordered local
Clifford equivalence.  Appendix~\ref{app:local-orbit-proof}
proves both equivalence and realization of every
\(R\in\mathfrak X(V)\).

For individual channels, the environmental coordinate change must
also be included.  If
\(S_2=(P\oplus Q)S_1(U\oplus T)^{-1}\), suitable Clifford lifts give
\begin{equation}
 \cN_{S_2,\sigma}
 =\Ad_{K_P}\circ
 \cN_{S_1,K_T^\dagger\sigma K_T}\circ\Ad_{K_U^\dagger}.
 \label{eq:family-equivalence}
\end{equation}
Choose the purifications consistently as
\[
 |\phi_{K_T^\dagger\sigma K_T}\rangle
 =(K_T^\dagger\otimes I_R)|\phi_\sigma\rangle.
\]
The full complementary channels then obey the same relation with
\(\Ad_{K_Q}\otimes\id_R\) at the output.  Thus the channel families agree after a fixed Clifford relabeling of
the environment.  Holding the environment fixed need not preserve this
channel equivalence.
Different interaction invariants need not give different channels
at a fixed environment.  A maximally mixed environment, for example,
gives the completely depolarizing channel for every regular interaction.

Through the Choi state, four-port Clifford equivalence is a case of
the stabilizer-state equivalence problem in
Ref.~\cite[Theorem~15]{EnglbrechtKraftKraus2022}.
For regular interactions, Appendix~\ref{app:local-orbit-proof}
identifies it with the symplectic double-coset classification of
Ref.~\cite[Proposition~3.1]{GoldsteinGuralnick2007}.

For one qudit per port the invariant is scalar: a \(2\times2\) matrix
satisfying \(R^{\mathsf T}J=JR\) is a scalar in odd characteristic.
There are then \(d-2\) regular ordered classes.
Their representatives and stabilizers under the local
symplectic action are given in
Ref.~\cite[Theorems~3.2 and~3.3]{Pahari2026}.
For \(n\ge2\), nonscalar cross-ratios occur and cannot be converted
to a BGJ componentwise coupling by port-local Clifford operations.

\subsection{Universal transport and convolution}

Within the classification above, positive convolution belongs to the
class where every environmental Clifford can be absorbed at the
other ports.  We characterize this class before studying its induced
channels.

Transport of every \(F\in\Sp(V)\) requires the conjugations in
Eq.~\eqref{eq:transport-compensators} to preserve \(\Sp(V)\).
Define the conformal symplectic group by
\begin{align*}
 \CSp(V):=\{M\in\operatorname{GL}(V):{}&
 M^{\mathsf T}JM=\mu(M)J,
 \\[-2pt]
 &\mu(M)\in\F^\times\}.
\end{align*}
These matrices need not be symplectic; their conjugations preserve
the symplectic group.

\begin{corollary}[Universal transport]
\label{cor:universal-transport}
For a regular action \(S\), the following are equivalent:
\begin{enumerate}
 \item every \(F\in\Sp(V)\) is transportable through \(S\);
 \item \(\mathscr R(S)=rI_V\) for some \(r\in\F\setminus\{0,1\}\);
 \item \(B,D,C^{-1}D\in\CSp(V)\).
\end{enumerate}
Every regular interaction with scalar cross-ratio is ordered
port-locally symplectically equivalent to a positive componentwise
convolution interaction.
\end{corollary}

The proof is given in Appendix~\ref{app:universal-proof}.

For the BGJ blocks in Eq.~\eqref{eq:convolution-blocks},
the cross-ratio is
\begin{equation}
 \mathscr R(S_G)=r(G)I,~
 r(G)=\frac{g_{00}g_{11}}{g_{01}g_{10}}\in\F\setminus\{0,1\}.
 \label{eq:scalar-rG}
\end{equation}

The scalar parameter also classifies positive componentwise couplings
under role-preserving multiplicative basis relabelings with the
environment fixed~\cite{XiongCoupling2026}.  Theorem~\ref{thm:local-orbit-classification}
allows independent Clifford operations at all four ports; the resulting
channel relation includes the environmental change in
Eq.~\eqref{eq:family-equivalence}.

\begin{corollary}[Environmental Clifford covariance for convolution]
\label{prop:environment-covariance}
Let \(G\) be positive and let \(K\) be an \(n\)-qudit Clifford.
There are Cliffords \(K_A,K_{A'},K_{B'}\) such that
\begin{equation*}
 U_G(I\otimes K)\doteqph
 (K_{A'}\otimes K_{B'})U_G(K_A\otimes I).
\end{equation*}
For \(\sigma'=K\sigma K^\dagger\), with the purifications chosen as in
Corollary~\ref{cor:general-pair},
\begin{align*}
 \Lambda_{G,\sigma'}&=\Ad_{K_{A'}}\circ
 \Lambda_{G,\sigma}\circ\Ad_{K_A},
 \\
 \widehat\Lambda_{G,\sigma'}&=(\Ad_{K_{B'}}\otimes\id_R)\circ
 \widehat\Lambda_{G,\sigma}\circ\Ad_{K_A}.
\end{align*}
\end{corollary}

\begin{proof}
The scalar cross-ratio \(\mathscr R(S_G)=r(G)I\) in
Eq.~\eqref{eq:scalar-rG} commutes with the homogeneous action of
every Clifford \(K\).  Corollary~\ref{cor:general-pair} therefore
gives the unitary transport identity and both channel relations.
\end{proof}

\subsection{Stabilizer environments and noiseless subsystems}
\label{sec:stabilizer-general}
\label{sec:heterogeneous}

The preceding classification concerns the interaction.  For a specified
pure stabilizer environment, transmission also depends on the position
of its stabilizer subspace relative to the cross-ratio.
Pure stabilizer environments give entanglement-breaking channels for
positive convolution~\cite[Theorem~1]{XiongKimLongWu2026}.
For a general regular interaction, the induced isometry is a
stabilizer-code isometry.  The existence of its noiseless, dephasing,
and depolarizing factors follows from the normal form of Looi and
Griffiths~\cite[Theorem~9]{LooiGriffiths2011}.
The next theorem computes their multiplicities, and hence the
capacities, from the cross-ratio and the environmental stabilizer
subspace.  The distinction is whether the Weyl directions retained
by the channel commute: noncommuting pairs give quantum degrees of
freedom that survive the interaction.

On one qudit, let \(\id_d\) be the identity channel,
\(\Delta_Z(X)=\sum_{j\in\F}\langle j|X|j\rangle\ketbra{j}{j}\)
complete computational-basis dephasing, and
\(\Omega_d(X)=\Tr(X)I/d\) the completely depolarizing channel.

\begin{theorem}[Cross-ratio rank and stabilizer-channel capacities]
\label{thm:stabilizer-regular}
Let \(S\) be regular, and let \(\psi_L\) be a pure stabilizer
environment with Lagrangian stabilizer subspace \(L\subset V\) and
arbitrary stabilizer character.  Let \(E\) be a \(2n\times n\)
matrix whose columns form a basis of \(L\), and set
\[
 q=\frac12\operatorname{rank}(E^{\mathsf T}J\mathscr R(S)E).
\]
Then \(q\) is an integer with \(0\le q\le\lfloor n/2\rfloor\),
and input and output Clifford conjugations transform the induced
receiver channel into
\begin{equation}
 \id_d^{\otimes q}\otimes
 \Delta_Z^{\otimes(n-2q)}\otimes\Omega_d^{\otimes q}.
 \label{eq:stabilizer-channel-normal}
\end{equation}
Its capacities are
\[
 Q=P=q\log_2d,~ C=(n-q)\log_2d.
\]
Equivalently, the integer $q$ is determined by
\[
 2q=n-\dim\bigl(L\cap\mathscr R(S)^{-1}L\bigr).
\]
\end{theorem}

The proof is given in Appendix~\ref{app:stabilizer-normal-proof}.
The relative position of \(L\) and its image under the
cross-ratio determines the number of noiseless qudits after Clifford
encoding and decoding.  In the normal form, \(q\) qudits are
transmitted unchanged, \(n-2q\) retain only a classical basis label,
and \(q\) are completely depolarized.  The interaction
orbit depends on the conjugacy class of \(\mathscr R\), whereas the
induced pure-stabilizer channel, up to input and output Cliffords,
depends on \((\mathscr R,L)\) through \(q\).

The normal form also identifies when the channel destroys all
input entanglement.  We first give the criterion for a specified
stabilizer environment, then ask when it holds for every such
environment.  The latter property characterizes universal transport.

\begin{corollary}[Entanglement breaking and universal transport]
\label{cor:universal-stabilizer-eb}
Let $S$ be regular and write $\mathscr R=\mathscr R(S)$.

\emph{(a) Fixed environment.}  For a pure stabilizer environment
$\psi_L$, with $q$ as in Theorem~\ref{thm:stabilizer-regular},
\begin{align*}
 \cN_{S,\psi_L}\text{ is EB}
 &\iff q=0\iff \mathscr R L=L\\
 &\iff Q(\cN_{S,\psi_L})=0.
\end{align*}

\emph{(b) All pure stabilizer environments.}  The following are equivalent:
\begin{enumerate}
 \item every environmental Clifford is transportable;
 \item \(\mathscr R(S)\) is scalar;
 \item \(\cN_{S,\psi_L}\) is entanglement breaking for every pure
       stabilizer environment \(\psi_L\);
 \item \(Q(\cN_{S,\psi_L})=0\) for every pure stabilizer
       environment \(\psi_L\).
\end{enumerate}
\end{corollary}

Appendix~\ref{app:universal-stabilizer-proof} first proves part~(a).
It then uses the fact that a linear map preserving every Lagrangian
subspace is scalar.  Thus every regular interaction with nonscalar
cross-ratio admits a pure stabilizer environment with positive
quantum capacity.

For two distinct scalar sectors, the rank formula of
Theorem~\ref{thm:stabilizer-regular} becomes an entanglement-entropy
formula: the quantum and private capacities equal the pure stabilizer
environment's entanglement between the sectors.

\begin{corollary}[Entanglement between two cross-ratio sectors]
\label{cor:two-sector-entanglement}
Let $S$ be regular and suppose
\[
 V=V_1\oplus^\perp V_2,~ \dim V_i=2n_i,
\]
with
\[
 \mathscr R(S)=r_1I_{V_1}\oplus r_2I_{V_2},~ r_1\ne r_2.
\]
Choose symplectic isomorphisms $T_i:\F^{2n_i}\to V_i$ and the
Clifford $U_T$ associated with
$T(z_1,z_2)=T_1z_1+T_2z_2$, as in
Eq.~\eqref{eq:sector-weyl-factorization}.
For a pure stabilizer state $\psi$, define
\[
 \widetilde\psi=U_T^\dagger\psi U_T,~
 \psi_{B_i}=\Tr_{B_{3-i}}\widetilde\psi.
\]
Then
\[
 Q(\cN_{S,\psi})=P(\cN_{S,\psi})
 =S(\psi_{B_1})=S(\psi_{B_2}).
\]
Thus each unit $\log_2d$ of entanglement across these environmental
subsystems gives one noiseless qudit in the channel normal form.
\end{corollary}

The proof is given in Appendix~\ref{app:two-sector-entanglement}.

The distinction $r_1\ne r_2$ is essential.  If $r_1=r_2$, then
$\mathscr R$ is scalar and Corollary~\ref{cor:universal-stabilizer-eb}
gives $Q=P=0$ for every pure stabilizer environment, even when it
is entangled across the chosen subsystems.

\begin{example}[A noiseless qudit from distinct cross-ratios]
\label{ex:two-sector-code}
Let $d=5$, with two qudits in each input register.  Apply the BGJ
coupling $G_i$ to the pair $A_iB_i$, where
\begin{equation}
 G_1=\begin{pmatrix}1&1\\1&2\end{pmatrix},
 ~ G_2=\begin{pmatrix}1&1\\1&3\end{pmatrix}.
 \label{eq:inhomogeneous-example}
\end{equation}
In the order $(p_1,p_2,q_1,q_2)$, the cross-ratio is
$\operatorname{diag}(2,3,2,3)$.  For the Bell stabilizer environment
\[
 |\Phi_5\rangle=5^{-1/2}\sum_{j\in\mathbb F_5}|j,j\rangle,
\]
Corollary~\ref{cor:two-sector-entanglement} gives $Q=P=\log_2 5$.
Theorem~\ref{thm:stabilizer-regular} identifies the whole channel,
up to input and output Cliffords, as $\id_5\otimes\Omega_5$.
For the same interaction with product environment $|0,0\rangle$,
each coordinate instead gives complete dephasing, up to basis
relabeling, and $Q=P=0$.

The noiseless qudit has a simple explicit code.  Encode
$|i\rangle\mapsto|i,i\rangle$ and decode the receiver's output by
\[
 |x,y\rangle\longmapsto|4x+2y,\,3x+2y\rangle.
\]
All arithmetic is in $\mathbb F_5$.  The resulting joint isometry is
\begin{align*}
 |i\rangle\longmapsto|i\rangle_{A'_1}\otimes|\chi_0\rangle,
~|\chi_0\rangle=\frac1{\sqrt5}\sum_k
 |k\rangle_{A'_2}|k,3k\rangle_{B'}.
\end{align*}
The second factor is independent of $i$, so the code preserves
arbitrary superpositions and leaks no logical information.
Appendix~\ref{app:heterogeneous} gives the coordinate calculation
and the full channel reduction.
\end{example}

Both environments in Example~\ref{ex:two-sector-code} have zero MRM.  Their different capacities show
why an environmental MRM bound requires a restriction on the
interaction; the next section treats the scalar-cross-ratio class.

\section{Environmental reduction and capacity bounds}
\label{sec:factorization}
\label{sec:environmental-bound}

We now restrict to regular interactions with scalar cross-ratio.
Corollary~\ref{cor:universal-transport} supplies a positive BGJ
representative in each such local orbit.  We work in these coordinates,
writing the interaction as \(U_G\); the corresponding environmental
Clifford change is understood as in Eq.~\eqref{eq:family-equivalence}.
Every environmental Clifford is then transportable.  This allows us
to separate the environment's stabilizer constraints from its remaining
state: the constrained coordinates give both outputs the same classical
label, while the remaining coordinates determine private and quantum
communication.

\subsection{Environmental normal form and capacity reduction}
\label{sec:capacities}

\begin{lemma}[Environmental Clifford normal form]
\label{lem:environment-normal-form}
If \(\dim S_\rho=\ell\), there is a Clifford unitary \(K\) such that
\begin{equation*}
 K\rho K^\dagger
 =|0\rangle\!\langle0|^{\otimes\ell}\otimes\tau,
 ~ m=n-\ell,
\end{equation*}
where \(S_\tau=\{0\}\).  In the same coordinates,
\begin{align*}
 K\cM(\rho)K^\dagger
 &=|0\rangle\!\langle0|^{\otimes\ell}\otimes I/d^m,
 \\
 \MRM(\rho)&=m\log_2d-S(\tau).
\end{align*}
\end{lemma}

This is the Clifford normal form for an isotropic Weyl subgroup and
its character~\cite{HostensDehaeneDeMoor2005,Gross2006}; a proof is given
in Appendix~\ref{app:normal-form}.
The following theorem applies it through transport to decompose the
channel and remove the shared label from the capacity optimization.
The shared-label entropy cancellation appears in the beam-splitter
analysis of Ref.~\cite[arXiv:2401.12105v3, Appendix C, Eqs. (C23)–(C30)]{BuJaffePRL2025}.

Write \(A=A_LA_E\) and \(B=B_LB_E\), where the \(L\) registers
contain the \(\ell\) constrained coordinates and the \(E\) registers
contain the \(m\) residual coordinates; use the same split at the outputs.
Superscripts \((s)\) display the number of qudits.  For
\(b\in\F^\times\), define the relabeled dephasing channel
\[
 \cD_b^{(s)}(X)
 =\sum_{x\in\F^s}\langle x|X|x\rangle\ketbra{bx}{bx}.
\]

\begin{theorem}[Stabilizer-sector and capacity reduction]
\label{thm:stabilizer-factorization}
\label{thm:capacity-reduction}
Let \(G\) be positive, and choose the normal form
\[
 K\sigma K^\dagger=\sigma_0
 =\ketbra{0^\ell}{0^\ell}\otimes\tau,
~ m=n-\ell.
\]
Put \(a=Ng_{11}\) and \(c=-Ng_{01}\), where \(N=(\det G)^{-1}\).
Then:

\emph{(i) Channel decomposition.}
Using the purification
\( |0^\ell\rangle_{B_L}\otimes|\phi_\tau\rangle_{B_ER}\),
\begin{align}
 \Lambda_{G,\sigma_0}^{(n)}
 =\cD_a^{(\ell)}\otimes\Lambda_{G,\tau}^{(m)},~
 \widehat\Lambda_{G,\sigma_0}^{(n)}
 =\cD_c^{(\ell)}\otimes\widehat\Lambda_{G,\tau}^{(m)}.
 \label{eq:complement-factorization}
\end{align}
The residual complementary output is \(B'_ER\).
The receiver--complementary pair for \(\sigma\) is equivalent to this
pair under a common input unitary and separate output unitaries.

\emph{(ii) Capacity reduction.}
Consequently,
\begin{align}
 Q(\Lambda_{G,\sigma}^{(n)})
 =Q(\Lambda_{G,\tau}^{(m)}),~P(\Lambda_{G,\sigma}^{(n)})
 =P(\Lambda_{G,\tau}^{(m)}).
 \label{eq:Q-reduction}
\end{align}
For \(m=0\), the residual system is one dimensional and its capacities
are zero.
\end{theorem}

Figure~\ref{fig:mechanism} illustrates the decomposition in
part~(i). The stabilized coordinates give the receiver and
the complementary output the same classical label, up to
reversible relabeling. Part~(ii) shows that removing this
shared label leaves the quantum and private capacities
unchanged.

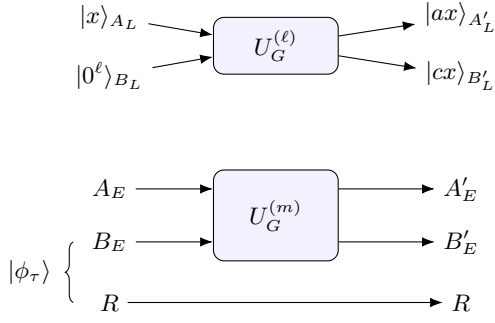
\begin{figure}[!tbp]
\centering
\begin{tikzpicture}[
  >=Latex,
  every node/.style={font=\small},
  box/.style={draw,rounded corners,minimum width=1.65cm,
              minimum height=0.75cm,fill=blue!5}]
 \node (x) at (0,0.72) {$\ket{x}_{A_L}$};
 \node (zero) at (0,-0.05) {$|0^{\ell}\rangle_{B_L}$};
 \node[box] (copy) at (2.2,0.34) {$U_G^{(\ell)}$};
 \node (ax) at (4.65,0.72) {$\ket{ax}_{A'_L}$ };
 \node (cx) at (4.65,-0.05) {$\ket{cx}_{B'_L}$ };
 \draw[->] (x)--(copy);
 \draw[->] (zero)--(copy);
 \draw[->] (copy)--(ax);
 \draw[->] (copy)--(cx);
\node (ae) at (0,-1.55) {$A_E$};
\node (be) at (0,-2.25) {$B_E$};

\node[box,minimum width=1.65cm,minimum height=1.2cm]
  (res) at (2.2,-1.90) {$U_G^{(m)}$};

\node (aout) at (4.65,-1.55) {$A'_E$};
\node (bout) at (4.65,-2.25) {$B'_E$};

\draw[->] (ae.east) -- (res.west |- ae);
\draw[->] (be.east) -- (res.west |- be);
\draw[->] (res.east |- aout) -- (aout.west);
\draw[->] (res.east |- bout) -- (bout.west);

\node (rin)  at (0,-3.05) {$R$};
\node (rout) at (4.65,-3.05) {$R$};
\draw[->] (rin.east) -- (rout.west);

\draw[decorate,decoration={brace,amplitude=3pt}]
  (-0.48,-3.05) -- (-0.48,-2.25)
  node[midway,left=5pt] {$|\phi_\tau\rangle$};
\end{tikzpicture}
\caption{Stabilizer-sector reduction for a positive BGJ coupling
$G$, in environmental Clifford normal coordinates.
On the stabilized coordinates,
$U_G^{(\ell)}|x,0^\ell\rangle=|ax,cx\rangle$,
with $x\in\mathbb F_d^\ell$, $a=(\det G)^{-1}g_{11}$,
and $c=-(\det G)^{-1}g_{01}$.
The two marginal channels are $\mathcal D_a^{(\ell)}$
and $\mathcal D_c^{(\ell)}$.
The remaining coordinates interact through $U_G^{(m)}$
with environment $\tau$.
The brace denotes a purification $|\phi_\tau\rangle_{B_E R}$;
$R$ bypasses the interaction.
The residual receiver and complementary outputs are $A'_E$
and $B'_E R$, respectively.}
\label{fig:mechanism}
\end{figure}

\begin{proof}
We first obtain the common dilation from transport and the normal form.
We then cancel the shared label in private and coherent information.

By Corollary~\ref{prop:environment-covariance}, it suffices to consider
\(\sigma_0=\ketbra{0^\ell}{0^\ell}\otimes\tau.\) Write \(A=A_LA_E\) and \(B'=B'_LB'_E\). An input state has the block decomposition
$\rho=\sum_{xy}\ket{x}\bra{y}\otimes R_{xy}$ with $\ket{x},\ket{y}\in\mathcal{H}_\ell$ and
\begin{align*}
\Lambda_{G,\sigma_0}^{(n)}(\rho)&=\sum_{x,y}\textmd{Tr}_{B'}[U^{(n)}_G\ket{x}\bra{y}\otimes R_{xy}\otimes\ket{0^{\ell}}\bra{0^{\ell}}\otimes\tau U^{(n)\dagger}_G]\\
&=\sum_x\ket{ax}\bra{ax}\otimes \textmd{Tr}_{B'_E}[U^{(m)}_G(R_{xx}\otimes\tau) U^{(m)\dagger}_G]\\
&=(\cD_a^{(\ell)}\otimes\Lambda_{G,\tau}^{(m)})(\rho).
\end{align*}


Let $|\phi_{\tau}\rangle$ be the purification of $\tau$ and 
\begin{align*}
W^{(\ell)}_{G,0}\ket{x}:=U^{(\ell)}_G(\ket{x}\otimes\ket{0^{\ell}})=\ket{ax}_{A'_L}\otimes\ket{cx}_{B'_L}.
\end{align*}
After regrouping the output registers, a Stinespring isometry is
\[
 V_{G,\sigma_0}=W_{G,0}^{(\ell)}\otimes V_{G,\tau},
\]
where $V_{G,\tau}|\varphi\rangle=(U_G^{(m)}\otimes I_R) (|\varphi\rangle\otimes|\phi_\tau\rangle)$ is a Stinespring isometry of $\Lambda_{G,\tau}^{(m)}$. Then,
\begin{align*}
 \widehat\Lambda_{G,\sigma_0}^{(n)}(\rho)=&\sum_x\ket{cx}\bra{cx}\otimes \textmd{Tr}_{A'_E}[V_{G,\tau}(R_{xx})V_{G,\tau}^{\dagger}]\\
 =&(\cD_c^{(\ell)}\otimes\widehat\Lambda_{G,\tau}^{(m)})(\rho),
\end{align*}
where the last line follows from the definition of $\widehat\Lambda_{G,\tau}^{(m)}$. This proves
\eqref{eq:complement-factorization}.


Fix \(r\ge1\) and write
\[
\cT=\Lambda_{G,\sigma_0}^{(n)},~\widehat\cT=\widehat\Lambda_{G,\sigma_0}^{(n)},~
 \cN=\Lambda_{G,\tau}^{(m)},~
 \widehat\cN=\widehat\Lambda_{G,\tau}^{(m)}.
\]
Write an arbitrary state \(\rho\) on \(\mathcal H_{rn}\) as $\rho=\sum_{xy}\ket{x}\bra{y}\otimes R_{xy}$, with $\ket{x},\ket{y}\in\mathcal{H}_{r\ell}$. Equation \eqref{eq:complement-factorization} gives
\begin{align*}
I_c(\rho,\cT^{\otimes r})=&S[\sum_xq_x\ket{ax}\bra{ax}\otimes\cN^{\otimes r}(\rho_x)]\\
&-S[\sum_xq_x\ket{cx}\bra{cx}\otimes\widehat{\cN}^{\otimes r}(\rho_x)]\\
=&\sum_x q_xI_c(\rho_x,\cN^{\otimes r}),
\end{align*}
where $q_x=\textmd{Tr}R_{xx}$ and $\rho_x=q^{-1}_xR_{xx}$. Zero-probability terms are omitted throughout the proof. 
The second equality uses the entropy formula for
block-diagonal states. Since $a,c\ne0$, the label
distributions at both outputs have the same Shannon
entropy, which cancels. Taking the maximum over $\rho$ implies that for each $r\ge1$,
\begin{align*}
 Q^{(1)}(\cT^{\otimes r})\le  Q^{(1)}(\cN^{\otimes r}). 
\end{align*}

Conversely, fix any \(x_0\in\mathbb F_d^{r\ell}\). Optimizing over states \(\rho_0\) on \(\mathcal H_{rm}\) gives
\begin{align*}
Q^{(1)}(\cT^{\otimes r})&=\max_{\rho}I_c(\rho,\cT^{\otimes r})\\
&\ge \max_{\rho_0}I_c(\ket{x_0}\bra{x_0}\otimes\rho_0,\cT^{\otimes r})\\
&= \max_{\rho_0}I_c(\rho_0,\cN^{\otimes r})\\
&=Q^{(1)}(\cN^{\otimes r}).
\end{align*}
Consequently, one has $Q^{(1)}(\cT^{\otimes r})=Q^{(1)}(\cN^{\otimes r})$ and the quantum capacity satisfies
\begin{align*}
 Q(\cT)&=\sup_{r\ge1}\frac{Q^{(1)}(\cT^{\otimes r})}{r}=\sup_{r\ge1}\frac{Q^{(1)}(\cN^{\otimes r})}{r}=Q(\cN),
\end{align*}
and Corollary \ref{prop:environment-covariance} implies
\begin{align*}
 Q(\Lambda_{G,\sigma}^{(n)})=Q(\Lambda_{G,\sigma_0}^{(n)})=Q(\Lambda_{G,\tau}^{(m)}).   
\end{align*}

We next consider an arbitrary ensemble \(\mathsf E=\{p_u,\rho_u\}\) on \(\mathcal H_{rn}\). Write \(\rho_u=\sum_{x,y}|x\rangle\langle y|\otimes R_{xy}^u\). The marginal probability of label \(x\) is \(q_x=\sum_u p_u\operatorname{Tr}R_{xx}^u\). For \(q_x>0\), the conditional ensemble is
 $\mathsf E_x=\{p^x_u,\rho^u_x\}$, where $p^x_u=\frac{p_uq^u_x}{q_x}$, $\rho^u_x=\frac{R^u_{xx}}{q^u_x}$ with $q^u_x=\Tr(R^u_{xx})$. 
 The conditional average state is $\bar\rho_x=\sum_u p_u^x\rho_x^u$. Hence
\begin{align}\label{eq:receiver-chain-rule}
\chi(\mathsf E;\cT^{\otimes r})=&S\!\left(\sum_up_u\cT^{\otimes r}(\rho_u)\right)
  -\sum_up_uS(\cT^{\otimes r}(\rho_u)) \nonumber\\
=&H(q)+\sum_xq_xS(\cN^{\otimes r}(\bar{\rho}_x))\nonumber\\
&-\sum_up_u\big[H(q_u)+\sum_xq^u_xS[\cN^{\otimes r}(\rho^u_x)]\big]\nonumber\\
=&H(q)-\sum_up_uH(q_u)+\sum_xq_x\Big[\chi(\mathsf E_x;\cN^{\otimes r})\big],
\end{align}
where $H(q)$ is the Shannon entropy of vector $q=(q_x)_x$ and $H(q_u)$ is the Shannon entropy of $(q^u_x)_x$. Similarly, it also holds that
\begin{align*}
\chi(\mathsf E;\widehat{\cT}^{\otimes r})=H(q)-\sum_up_uH(q_u)+\sum_xq_x\Big[\chi(\mathsf E_x;\widehat{\cN}^{\otimes r})\big],
\end{align*}
therefore,
\begin{align*}
 &\chi(\mathsf E;\cT^{\otimes r})-\chi(\mathsf E;\widehat{\cT}^{\otimes r})\\
 =&\sum_xq_x[\chi(\mathsf E_x;\cN^{\otimes r})-\chi(\mathsf E_x;\widehat{\cN}^{\otimes r})]\\
 \le& \sum_xq_x\sup_{\mathsf F}[\chi(\mathsf F;\cN^{\otimes r})-\chi(\mathsf F;\widehat{\cN}^{\otimes r})]\\
 =& P^{(1)}(\cN^{\otimes r}).
\end{align*}
Taking the supremum over all ensembles gives $P^{(1)}(\cT^{\otimes r})\le P^{(1)}(\cN^{\otimes r})$.

Conversely, for an arbitrary ensemble on the residual input $\mathsf F=\{p_u,\omega_u\}$, define the ensemble $\mathsf E_{x_0}=\{p_u,\ket{x_0}\bra{x_0}\otimes\omega_u\}$. Equation \eqref{eq:complement-factorization} gives
\begin{align*}
\chi(\mathsf E_{x_0};\cT^{\otimes r})=\chi(\mathsf F;\cN^{\otimes r}),~\chi(\mathsf E_{x_0};\widehat{\cT}^{\otimes r})=\chi(\mathsf F;\widehat{\cN}^{\otimes r}),
\end{align*}
and
\begin{align*}
P^{(1)}(\cT^{\otimes r})&=\sup_{\mathsf E}[\chi(\mathsf E;\cT^{\otimes r})-\chi(\mathsf E;\widehat{\cT}^{\otimes r})]\\
&\ge\sup_{\mathsf E_{x_0}}[\chi(\mathsf E_{x_0};\cT^{\otimes r})-\chi(\mathsf E_{x_0};\widehat{\cT}^{\otimes r})]\\
&=\sup_{\mathsf F}[\chi(\mathsf F;\cN^{\otimes r})-\chi(\mathsf F;\widehat{\cN}^{\otimes r})]=P^{(1)}(\cN^{\otimes r}).
\end{align*}
Consequently, it holds that
\begin{align*}
 P^{(1)}(\cT^{\otimes r})=P^{(1)}(\cN^{\otimes r}),
\end{align*}
and the private capacity satisfies
\begin{align*}
 P(\cT)&=\sup_{r\ge1}\frac{P^{(1)}(\cT^{\otimes r})}{r}=\sup_{r\ge1}\frac{P^{(1)}(\cN^{\otimes r})}{r}=P(\cN).
\end{align*}
In conclusion, Corollary \ref{prop:environment-covariance} implies
\begin{align*}
 P(\Lambda_{G,\sigma}^{(n)})=P(\Lambda_{G,\sigma_0}^{(n)})=P(\Lambda_{G,\tau}^{(m)}). 
\end{align*}
\end{proof}

The removed coordinates can still carry classical information.
Equation~\eqref{eq:receiver-chain-rule} bounds the \(r\)-use Holevo
information by \(r\ell\log_2d+\chi^*(\cN^{\otimes r})\). Conversely, adjoining an independent uniform label and optimizing over residual ensembles gives the reverse bound. Thus
\begin{equation}
 C(\Lambda_{G,\sigma}^{(n)})
 =\ell\log_2d+C(\Lambda_{G,\tau}^{(m)}).
 \label{eq:classical-reduction}
\end{equation}
It also follows from the Holevo additivity theorem for EB
channels~\cite{Shor2002}.

\subsection{The environmental MRM bound}

The capacity reduction removes the shared label.  To bound the
remaining channel, we use the output entropy estimate in
Proposition~\ref{prop:basic-identities}.  It applies to the full block
of uses and therefore includes inputs entangled across uses.

\begin{theorem}[Environmental MRM bound]
\label{thm:environmental-bound}
Let \(S\) be a regular Clifford interaction with scalar cross-ratio.
For every pure or mixed environment state \(\sigma\),
\begin{equation*}
 Q(\cN_{S,\sigma})\le P(\cN_{S,\sigma})\le\MRM(\sigma).
\end{equation*}
In particular, every positive BGJ coupling \(G\) satisfies
\begin{equation*}
 Q(\Lambda_{G,\sigma})\le P(\Lambda_{G,\sigma})\le\MRM(\sigma).
\end{equation*}
\end{theorem}

\begin{proof}
By Corollary~\ref{cor:universal-transport} and
Eq.~\eqref{eq:family-equivalence}, we may pass to a positive
BGJ representative with a Clifford-transformed environment.
This preserves both channel capacities and environmental MRM.
Theorem~\ref{thm:capacity-reduction} and
Lemma~\ref{lem:environment-normal-form} then give a residual
channel \(\cN=\Lambda_{G,\tau}^{(m)}\) satisfying
\[
 P(\cN_{S,\sigma})=P(\cN),
~
 \MRM(\sigma)=m\log_2d-S(\tau).
\]
If \(m=0\), the residual channel is one dimensional and the
claim is immediate. Assume \(m\ge1\).

Fix \(r\ge1\), put \(\cN_r=\cN^{\otimes r}\), and let
\(\mathsf E=\{p_u,\rho_u\}\) be an arbitrary ensemble on its
\(rm\)-qudit input, with
\(\overline\rho=\sum_up_u\rho_u\).
Applying Eq.~\eqref{eq:entropy-increase} to the full
\(r\)-use interaction gives
\[
 S(\cN_r(\rho_u))
 \ge S(\tau^{\otimes r})
 =rS(\tau),
\]
including when \(\rho_u\) is entangled across uses.
For a complementary channel \(\widehat\cN_r\), we therefore have
\begin{align*}
 &\chi(\mathsf E;\cN_r)
   -\chi(\mathsf E;\widehat\cN_r)\\
 &\quad\le \chi(\mathsf E;\cN_r)\\
 &\quad=S(\cN_r(\overline\rho))
       -\sum_up_uS(\cN_r(\rho_u))\\
 &\quad\le rm\log_2d-\sum_up_u\,rS(\tau)\\
 &\quad=r\bigl[m\log_2d-S(\tau)\bigr].
\end{align*}
The first inequality uses nonnegativity of the complementary
Holevo information; the second uses the output dimension
\(d^{rm}\) and the entropy bound above.

Taking the supremum over input ensembles, dividing by \(r\),
and regularizing yields
\begin{align*}
 Q(\cN_{S,\sigma})
 &\le P(\cN_{S,\sigma})\\
 &=P(\cN)
   =\sup_{r\ge1}\frac1rP^{(1)}(\cN^{\otimes r})\\
 &\le m\log_2d-S(\tau)
   =\MRM(\sigma).
\end{align*}
\end{proof}

The scalar condition is also necessary if the bound is to hold
uniformly over environmental states.  This gives a capacity
characterization of the interactions identified in
Corollary~\ref{cor:universal-transport}.

\begin{corollary}[Universal environmental MRM bounds]
\label{cor:universal-mrm-bound}
Fix a regular interaction \(S\), and let \(\sigma\) range over all
environment states.  The following are equivalent:
\begin{enumerate}
 \item \(\mathscr R(S)=rI_V\) for some
       \(r\in\F\setminus\{0,1\}\);
 \item \(P(\cN_{S,\sigma})\le\MRM(\sigma)\) for all \(\sigma\);
 \item \(Q(\cN_{S,\sigma})\le\MRM(\sigma)\) for all \(\sigma\).
\end{enumerate}
\end{corollary}

\begin{proof}
Theorem~\ref{thm:environmental-bound} gives (1)$\Rightarrow$(2),
and \(Q\le P\) gives (2)$\Rightarrow$(3).
For (3)$\Rightarrow$(1), take any pure stabilizer environment
\(\psi_L\).  Since \(\MRM(\psi_L)=0\), condition (3) implies
\(Q(\cN_{S,\psi_L})=0\) for every such environment.
Corollary~\ref{cor:universal-stabilizer-eb} then forces
\(\mathscr R(S)\) to be scalar.
\end{proof}

For scalar cross-ratios and a pure environment, the residual state is pure and
$\MRM(\sigma)=m\log_2d$.  The bound then coincides with the input-dimension
bound for the residual channel; identifying that channel is the role of
Theorem~\ref{thm:capacity-reduction}.  If there are no nontrivial exact
stabilizers, $m=n$ and this is the original input-dimension bound.
For mixed environments, subtracting $S(\tau)$ gives an additional
restriction.

When MRM vanishes, the residual environment is maximally mixed and
the factorization determines the channel explicitly.  We display this
normal form in the BGJ coordinates used for the factorization; the
same form holds, up to input and output Cliffords, throughout the
scalar-cross-ratio class.

\begin{corollary}[MSPS channel normal form]
\label{cor:msps-channel}
Let \(G\) be positive, and set \(\ell=\dim S_\sigma\),
\(m=n-\ell\), and \(a=(\det G)^{-1}g_{11}\).
If \(\MRM(\sigma)=0\), then \(\sigma=\cM(\sigma)\) and the induced channel is,
up to input and output Clifford conjugations,
\begin{equation}
 \cD_a^{(\ell)}\otimes\Omega_d^{\otimes m}.
 \label{eq:msps-channel-form}
\end{equation}
It is entanglement breaking and
\begin{align*}
 C(\Lambda_{G,\sigma})&=\ell\log_2d
 =n\log_2d-S(\sigma),
 \\
 P(\Lambda_{G,\sigma})&=Q(\Lambda_{G,\sigma})=0.
\end{align*}
\end{corollary}

\begin{proof}
The condition \(\MRM(\sigma)=0\) is equivalent to
\(\sigma=\cM(\sigma)\).  In the normal coordinates of
Lemma~\ref{lem:environment-normal-form}, this means
\(\tau=I/d^m\).  Equations~\eqref{eq:complement-factorization} and
\eqref{eq:depolarizing-example} then give
Eq.~\eqref{eq:msps-channel-form}.  For an input state $\rho$, let
$p_x=\Tr[(|x\rangle\langle x|\otimes I)\rho]$.
The normal form acts as
\[
 \rho\longmapsto\sum_xp_x|ax\rangle\langle ax|\otimes I/d^m.
\]
It measures a basis label and prepares an output state, so it is
entanglement breaking.
The private and quantum formulas follow from
Theorem~\ref{thm:capacity-reduction}, and the classical formula follows
from Eq.~\eqref{eq:classical-reduction}.
\end{proof}

\section{Discussion}
\label{sec:discussion}

For the regular Clifford interactions studied here, the cross-ratio
describes both the environmental operations compatible with the coupling
and its equivalence under changes of port coordinates.  Together with
the stabilizer subspace of a pure stabilizer environment, this invariant determines
the channel normal form.  Conjugate pairs of retained Weyl directions
support noiseless quantum transmission; directions commuting with the
whole retained subspace carry only classical labels.  The remaining
factors are completely depolarized.  The rank formula
and the associated Clifford encoding make this decomposition explicit
from the interaction and environment, without optimizing over input
states.

The scalar-cross-ratio class is precisely the class locally equivalent
to positive BGJ convolution.  Universal environmental transport in
this class explains both entanglement breaking for every pure
stabilizer environment and capacity reduction for arbitrary mixed
environments.  
This reduction removes a shared classical label while
preserving the quantum and private capacities exactly,
and yields the environmental MRM bound.
For non-scalar cross-ratios, the pure
stabilizer examples show how correlations within the environment
can instead supply noiseless degrees of freedom.

For mixed environments at non-scalar cross-ratio, the corresponding
reduction remains open.  A Clifford that puts the environment into
stabilizer normal form need not be transportable.  The centralizer
criterion specifies the available transformations; determining which
constraints they can separate from the transmitted quantum degrees
of freedom is a direct next problem.

\acknowledgments

This work is supported by the National Natural Science Foundation of China (Grant No. 12201555), the Natural Science Foundation of Hunan Province (Grant No. 2025JJ50050), and the Hunan Basic Science Research Center for Mathematical Analysis (2024JC2002). 

During manuscript preparation, the authors used \mbox{OpenAI} ChatGPT (GPT-6 Astra) to assist with language editing,
checking, and reformulation of some derivations. All AI-assisted mathematical content was independently verified and
revised by the authors, who take full responsibility for the results and the final manuscript.


\section*{Data availability}
The derivations and explicit examples supporting the conclusions are
contained in the article.  No external dataset was used.

\appendix

\section{Proofs for transport and local-orbit results}
\label{app:transport-proofs}

All matrices act on \(V=\F^{2n}\), except for the displayed two-port
block matrices; transposes and inverses are taken over \(\F\).

\subsection{Regular-block identities}

\begin{lemma}[Regular-block identities]
\label{lem:regular-block-identities}
Let
\(
S=\begin{psmallmatrix}A&B\\C&D\end{psmallmatrix}
\in\Sp(V\oplus V)
\)
and assume that \(B,C,D\) are invertible.  Put
\begin{equation*}
 M=C^{-1}D,
~
 \mathscr R=B^{-1}AC^{-1}D.
\end{equation*}
Then \(A\), \(\mathscr R\), and \(I-\mathscr R\) are invertible,
\begin{equation}
 \mathscr R^{\mathsf T}J=J\mathscr R,
 \label{eq:R-self-adjoint}
\end{equation}
and
\begin{align}
 B^{\mathsf T}JB&=J(I-\mathscr R)^{-1},
 \label{eq:HB-R}\\
 D^{\mathsf T}JD&=-J\mathscr R(I-\mathscr R)^{-1},
 \label{eq:HD-R}\\
 M^{\mathsf T}JM&=-J\mathscr R.
 \label{eq:HL-R}
\end{align}
Moreover,
\begin{equation}
 A=B\mathscr RM^{-1},
~ C=DM^{-1}.
 \label{eq:A-C-from-R}
\end{equation}
\end{lemma}

\begin{proof}
We use the symplectic block equations to determine the forms associated
with \(B,D,M\).  Their nonsingularity will also give the claimed
invertibility of \(A\) and \(\mathscr R\).

\emph{1. Determine the form associated with \(B\).}
Expanding \(S^{\mathsf T}(J\oplus J)S=J\oplus J\) gives
\begin{align*}
 A^{\mathsf T}JA+C^{\mathsf T}JC&=J,\\
 A^{\mathsf T}JB+C^{\mathsf T}JD&=0,\\
 B^{\mathsf T}JB+D^{\mathsf T}JD&=J.
\end{align*}
The other off-diagonal block is the negative transpose of the second
equation, so these three equations are also sufficient for
symplecticity.  Write
\(H_B=B^{\mathsf T}JB\) and \(H_D=D^{\mathsf T}JD\).
Both are nonsingular alternating forms.  The definitions of \(M\) and
\(\mathscr R\) give Eq.~\eqref{eq:A-C-from-R} without yet assuming
that \(A\) or \(\mathscr R\) is invertible.  Substituting these
expressions into the cross-block equation and multiplying on the left
by \(M^{\mathsf T}\) gives
\begin{equation*}
 H_B+H_D=J,
~
 \mathscr R^{\mathsf T}H_B+H_D=0.
\end{equation*}
Subtracting the two identities yields
\begin{equation*}
 (I-\mathscr R^{\mathsf T})H_B=J.
\end{equation*}
The matrices \(H_B\) and \(J\) are invertible.  Thus
\(I-\mathscr R\) is invertible and
\begin{equation}
 H_B=(I-\mathscr R^{\mathsf T})^{-1}J.
 \label{eq:HB-first-form}
\end{equation}

\emph{2. Derive symplectic self-adjointness of \(\mathscr R\).}
The expression for \(H_B\) must be alternating.  Transposing
Eq.~\eqref{eq:HB-first-form} and using \(J^{\mathsf T}=-J\) gives
\[
 H_B^{\mathsf T}=-J(I-\mathscr R)^{-1}.
\]
Since \(H_B^{\mathsf T}=-H_B\), we therefore have
\[
 (I-\mathscr R^{\mathsf T})^{-1}J
 =J(I-\mathscr R)^{-1}.
\]
Multiplication on the left by \(I-\mathscr R^{\mathsf T}\) and on
the right by \(I-\mathscr R\) gives
\[
 J(I-\mathscr R)=(I-\mathscr R^{\mathsf T})J,
\]
which is Eq.~\eqref{eq:R-self-adjoint}.  The preceding equality for
\(H_B\) proves Eq.~\eqref{eq:HB-R}.

\emph{3. Recover the remaining forms and invertibility.}
The identity \(H_D=J-H_B\) gives
\[
 H_D=J\bigl[I-(I-\mathscr R)^{-1}\bigr]
     =-J\mathscr R(I-\mathscr R)^{-1},
\]
proving Eq.~\eqref{eq:HD-R}.  Its left-hand side is nonsingular, so
\(\mathscr R\) is nonsingular.  Equation~\eqref{eq:A-C-from-R}
now also proves that \(A\) is invertible.

Finally, substitute Eq.~\eqref{eq:A-C-from-R} into the first-column
identity and multiply by \(M^{\mathsf T}\) and \(M\).  This gives
\begin{align*}
 M^{\mathsf T}JM
 &=\mathscr R^{\mathsf T}H_B\mathscr R+H_D\\
 &=J\mathscr R^2(I-\mathscr R)^{-1}
   -J\mathscr R(I-\mathscr R)^{-1}\\
 &=-J\mathscr R.
\end{align*}
Here \(\mathscr R^{\mathsf T}J=J\mathscr R\) was used to move
the transpose, and \(\mathscr R\) commutes with
\((I-\mathscr R)^{-1}\).  This proves Eq.~\eqref{eq:HL-R}.
\end{proof}

\subsection{The centralizer criterion}
\label{app:centralizer-proof}

\begin{proof}[Proof of Theorem~\ref{thm:centralizer-transport}]
Fix regular \(S\) and \(F\in\Sp(V)\), and write
\(\mathscr R=\mathscr R(S)\).
We first solve the block equations for the only possible compensators.
We then show that the commutation condition makes these candidates
symplectic, establishing sufficiency.

\emph{1. Necessary candidates and the commutation condition.}
Block multiplication rewrites Eq.~\eqref{eq:general-transport} as
\[
 \begin{pmatrix}A&BF\\C&DF\end{pmatrix}
 =\begin{pmatrix}
 X_FAZ_F&X_FB\\Y_FCZ_F&Y_FD
 \end{pmatrix},
\]
so
\begin{align*}
 A&=X_FAZ_F,~BF=X_FB,\\
 C&=Y_FCZ_F,~DF=Y_FD.
\end{align*}
The second, fourth, and third equations successively give
\begin{align*}
 X_F&=BFB^{-1},&Y_F&=DFD^{-1},\\
 Z_F&=C^{-1}Y_F^{-1}C=MF^{-1}M^{-1},&M&=C^{-1}D.
\end{align*}
Using \(A=B\mathscr RM^{-1}\), the remaining equation becomes
\begin{align*}
 X_FAZ_F&=BF\mathscr RF^{-1}M^{-1},\\
 A=X_FAZ_F
 &\quad\Longleftrightarrow\quad
 \mathscr R=F\mathscr RF^{-1}\\
 &\quad\Longleftrightarrow\quad
 F\mathscr R=\mathscr RF.
\end{align*}
This proves necessity and uniqueness of the candidate compensators.

\emph{2. Symplecticity of the candidates.}
Assume \(F\mathscr R=\mathscr RF\) and use the three candidates
just derived.  The block equations already reduce to this commutation
condition; the remaining requirement is
\(X_F,Y_F,Z_F\in\Sp(V)\).
By their conjugation formulas, it suffices to show that \(F\) preserves
the forms associated with \(B,D,M\).
Commutation with an invertible
matrix implies commutation with its inverse; hence
\begin{gather*}
 F(I-\mathscr R)=(I-\mathscr R)F\\
 \Longrightarrow\quad
 F(I-\mathscr R)^{-1}=(I-\mathscr R)^{-1}F.
\end{gather*}
For each of the three expressions
\[
 g(\mathscr R)\in
 \bigl\{(I-\mathscr R)^{-1},
 -\mathscr R(I-\mathscr R)^{-1},-\mathscr R\bigr\},
\]
we therefore have
\[
 F^{\mathsf T}Jg(\mathscr R)F
 =F^{\mathsf T}JFg(\mathscr R)=Jg(\mathscr R).
\]
By Eqs.~\eqref{eq:HB-R}--\eqref{eq:HL-R}, this says
\begin{gather*}
 F^{\mathsf T}HF=H,\\
 H\in\{B^{\mathsf T}JB,D^{\mathsf T}JD,M^{\mathsf T}JM\}.
\end{gather*}
Multiplication by \(F^{-\mathsf T}\) and \(F^{-1}\) also gives
\(F^{-\mathsf T}HF^{-1}=H\).

For invertible \(N\), if \(F\) preserves \(N^{\mathsf T}JN\), meaning
\(F^{\mathsf T}(N^{\mathsf T}JN)F=N^{\mathsf T}JN\), then its
conjugate by \(N\) is symplectic:
\begin{align*}
 (NFN^{-1})^{\mathsf T}J(NFN^{-1})
 &=N^{-\mathsf T}F^{\mathsf T}(N^{\mathsf T}JN)FN^{-1}\\
 &=N^{-\mathsf T}(N^{\mathsf T}JN)N^{-1}=J.
\end{align*}
Taking \(N=B\) and \(N=D\) proves symplecticity of \(X_F\) and
\(Y_F\).  Taking \(N=M\) and replacing \(F\) by \(F^{-1}\)
proves it for \(Z_F\).  Finally,
\begin{align*}
 X_FB&=BF,\qquad Y_FD=DF,\\
 Y_FCZ_F&=DF(M^{-1}M)F^{-1}M^{-1}=C,\\
 X_FAZ_F&=BF\mathscr RF^{-1}M^{-1}=A,
\end{align*}
which proves the transport identity.
\end{proof}

\subsection{Proof of Corollary~\ref{cor:general-pair}}
\label{app:clifford-lifts}

\begin{proof}[Proof of Corollary~\ref{cor:general-pair}]
The matrix criterion determines the homogeneous actions.  To lift it
for an arbitrary environmental Clifford, we first separate its Weyl
displacement and then track the same identity through both partial traces.

\emph{1. Separate the displacement.}
Let $K$ have homogeneous action $F$, and let $U_F$ be the exact lift
from Sec.~\ref{sec:weyl-clifford}.  Then $Q_0=KU_F^\dagger$ satisfies
\[
 Q_0w(z)Q_0^\dagger=\chi(z)w(z).
\]
The Weyl multiplication law implies
$\chi(z+z')=\chi(z)\chi(z')$.  Every such character has the form
$\chi(z)=\zeta^{[v,z]}$ for a unique $v\in V$, because the symplectic
form is nondegenerate.  Conjugation by $w(v)$ has exactly this action,
so $w(v)^\dagger Q_0$ commutes with the full Weyl basis and is a scalar.
Thus
\[
 K\doteqph w(v)U_F.
\]
The argument also applies on $V\oplus V$: any other Clifford lift of
$S$ differs from $U_S$ by a product of output Weyl operators and a phase.

\emph{2. Lift the transport identity.}
Necessity follows by taking homogeneous actions in
Eq.~\eqref{eq:general-unitary-transport} and applying
Theorem~\ref{thm:centralizer-transport}.
Conversely, assume $F\mathscr R(S)=\mathscr R(S)F$.
The exact lifts of the compensators in
Eq.~\eqref{eq:transport-compensators} satisfy
\[
 U_S(I\otimes U_F)\doteqph
 (U_{X_F}\otimes U_{Y_F})U_S(U_{Z_F}\otimes I),
\]
because both sides induce the same exact action on every joint Weyl
operator.  The environmental displacement is transported explicitly:
\[
 U_S(I\otimes w(v))=(w(Bv)\otimes w(Dv))U_S.
\]
Combining these identities proves
Eq.~\eqref{eq:general-unitary-transport} with
\begin{align*}
 K_A=U_{Z_F},~
 K_{A'}=w(Bv)U_{X_F},~
 K_{B'}=w(Dv)U_{Y_F}.
\end{align*}

\emph{3. Obtain the receiver and complementary channels.}
For $\sigma'=K\sigma K^\dagger$, choose the purification
$|\phi_{\sigma'}\rangle=(K\otimes I_R)|\phi_\sigma\rangle$.
Applying the unitary identity to an arbitrary input vector gives
\begin{align*}
 V_{S,\sigma'}|\psi\rangle
 &=(U_S(I\otimes K)\otimes I_R)
 (|\psi\rangle\otimes|\phi_\sigma\rangle)\\
 &\doteqph(K_{A'}\otimes K_{B'}\otimes I_R)
 V_{S,\sigma}K_A|\psi\rangle.
\end{align*}
The phase cancels on density operators.  Tracing out $B'R$ and $A'$,
respectively, therefore yields
\begin{align*}
 \cN_{S,\sigma'}
 &=\Ad_{K_{A'}}\circ\cN_{S,\sigma}\circ\Ad_{K_A},\\
 \widehat\cN_{S,\sigma'}
 &=\Ad_{K_{B'}\otimes I_R}\circ
 \widehat\cN_{S,\sigma}\circ\Ad_{K_A}.
\end{align*}
Both relations use the same input unitary and leave the purification
reference unchanged.
\end{proof}

\subsection{Completeness of the local-orbit invariant}
\label{app:local-orbit-proof}

\begin{proof}[Proof of Theorem~\ref{thm:local-orbit-classification}]
Lemma~\ref{lem:regular-block-identities} gives
\(\mathscr R(S)\in\mathfrak X(V)\) for every regular \(S\).
We prove that two regular interactions are locally equivalent exactly
when their cross-ratios are symplectically conjugate, and then realize
every \(R\in\mathfrak X(V)\).
Write
\[
 S=\begin{pmatrix}A&B\\C&D\end{pmatrix},~
 S'=\begin{pmatrix}A'&B'\\C'&D'\end{pmatrix}.
\]

\emph{1. Local equivalence implies conjugate cross-ratios.}
Suppose \(S'=(P\oplus Q)S(U\oplus T)^{-1}\), with
\(P,Q,U,T\in\Sp(V)\).  In block form this means
\[
 \begin{pmatrix}A'&B'\\C'&D'\end{pmatrix}
 =\begin{pmatrix}
 PAU^{-1}&PBT^{-1}\\QCU^{-1}&QDT^{-1}
 \end{pmatrix}.
\]
Comparing the four blocks and substituting into the cross-ratio gives
\begin{align*}
 \mathscr R(S')
 &=(TB^{-1}P^{-1})(PAU^{-1})\\
 &\quad\times(UC^{-1}Q^{-1})(QDT^{-1})\\
 &=TB^{-1}AC^{-1}DT^{-1}
  =T\mathscr R(S)T^{-1}.
\end{align*}
Thus the map on local orbits is well defined.

\emph{2. A conjugating matrix determines a local equivalence.}
Now assume \(\mathscr R'=T\mathscr RT^{-1}\), where
\(\mathscr R=\mathscr R(S)\), \(\mathscr R'=\mathscr R(S')\),
and \(T\in\Sp(V)\) is given.
We must construct \(P,Q,U\in\Sp(V)\) satisfying the displayed
block identity.  Put \(M=C^{-1}D\) and \(M'=C'^{-1}D'\).
The \(B'\) and \(D'\) blocks force
\[
 P=B'TB^{-1},~ Q=D'TD^{-1}.
\]
The \(C'\) block then forces
\[
 U=C'^{-1}QC=C'^{-1}D'TD^{-1}C=M'TM^{-1}.
\]
These choices satisfy three block equations.  It remains to prove that
\(P,Q,U\) are symplectic and that the \(A'\) block also agrees.

To verify symplecticity, we compare the alternating forms associated
with \(B,D,M\) and \(B',D',M'\).  The three expressions in
Eqs.~\eqref{eq:HB-R}--\eqref{eq:HL-R} are
\[
 g(X)\in\{(I-X)^{-1},-X(I-X)^{-1},-X\}.
\]
For each of them, conjugacy of \(\mathscr R,\mathscr R'\) and
symplecticity of \(T\) give
\begin{gather*}
 g(\mathscr R')=Tg(\mathscr R)T^{-1},~
 JT=T^{-\mathsf T}J,\\
 Jg(\mathscr R')=T^{-\mathsf T}Jg(\mathscr R)T^{-1}.
\end{gather*}
Lemma~\ref{lem:regular-block-identities} therefore yields
\begin{align*}
 B'^{\mathsf T}JB'&=T^{-\mathsf T}B^{\mathsf T}JBT^{-1},\\
 D'^{\mathsf T}JD'&=T^{-\mathsf T}D^{\mathsf T}JDT^{-1},\\
 M'^{\mathsf T}JM'&=T^{-\mathsf T}M^{\mathsf T}JMT^{-1}.
\end{align*}
Using the first identity, we obtain
\begin{align*}
 P^{\mathsf T}JP
 &=B^{-\mathsf T}T^{\mathsf T}(B'^{\mathsf T}JB')TB^{-1}\\
 &=B^{-\mathsf T}(B^{\mathsf T}JB)B^{-1}=J.
\end{align*}
The second and third identities give, in the same way,
\(Q^{\mathsf T}JQ=J\) and \(U^{\mathsf T}JU=J\).
Thus all three candidates are symplectic.

For the remaining block, Eq.~\eqref{eq:A-C-from-R} gives
\(A=B\mathscr RM^{-1}\) and
\(A'=B'\mathscr R'M'^{-1}\).  Hence
\begin{align*}
 PAU^{-1}
 &=B'TB^{-1}(B\mathscr RM^{-1})(MT^{-1}M'^{-1})\\
 &=B'T\mathscr RT^{-1}M'^{-1}
  =B'\mathscr R'M'^{-1}=A'.
\end{align*}
All four block equations now hold, proving
\(S'=(P\oplus Q)S(U\oplus T)^{-1}\).
This establishes injectivity on local orbits.

\emph{3. Every admissible cross-ratio is realized.}
Let \(R\in\mathfrak X(V)\) be prescribed.  We seek a regular
symplectic matrix \(S\) with \(\mathscr R(S)=R\).
We reverse the identities of Lemma~\ref{lem:regular-block-identities}:
first realize the required alternating forms, then assemble and verify
the four blocks.  The required forms are
\begin{align}
 H_B=J(I-R)^{-1},~
 H_D=-JR(I-R)^{-1},~
 H_M=-JR.
 \label{eq:forms-from-R}
\end{align}
To realize them as \(N^{\mathsf T}JN\), we must check that they
are alternating and nonsingular.  The hypothesis
\(R^{\mathsf T}J=JR\) gives
\begin{align*}
 (I-R^{\mathsf T})J&=J(I-R),\\
 (I-R^{\mathsf T})^{-1}J&=J(I-R)^{-1}.
\end{align*}
Consequently each of the three expressions \(g(R)\) above satisfies
\[
 g(R)^{\mathsf T}J=Jg(R),~
 [Jg(R)]^{\mathsf T}=-Jg(R).
\]
Since the characteristic is odd and \(J,R,I-R\) are invertible,
the forms in Eq.~\eqref{eq:forms-from-R} are alternating and
nonsingular.  A nonsingular alternating form $H$ admits a symplectic
basis: choose $e,f$ with $e^{\mathsf T}Hf=1$, then repeat on their
$H$-orthogonal complement.  With the paired basis vectors as the
columns of $W_H$, ordered with all $e$ vectors before the $f$ vectors,
we obtain
\[
 W_H^{\mathsf T}HW_H=J
 \quad\Longrightarrow\quad
 (W_H^{-1})^{\mathsf T}J W_H^{-1}=H.
\]
Applying this to $H_B,H_D,H_M$ supplies invertible
$B=W_{H_B}^{-1}$, $D=W_{H_D}^{-1}$, and $M=W_{H_M}^{-1}$ satisfying
\begin{gather*}
 B^{\mathsf T}JB=H_B,~ D^{\mathsf T}JD=H_D,~
 M^{\mathsf T}JM=H_M.
\end{gather*}
Following Eq.~\eqref{eq:A-C-from-R}, define
\[
 C=DM^{-1},~ A=BRM^{-1},~
 S=\begin{pmatrix}A&B\\C&D\end{pmatrix}.
\]
All four blocks are invertible.  To check that \(S\) is symplectic,
we use the three block conditions from
Lemma~\ref{lem:regular-block-identities}.  The prescribed forms obey
\begin{align*}
 H_B+H_D&=J,\\
 R^{\mathsf T}H_B+H_D&=0,\\
 R^{\mathsf T}H_BR+H_D&=-JR=H_M.
\end{align*}
Thus
\begin{align*}
 B^{\mathsf T}JB+D^{\mathsf T}JD&=J,\\
 A^{\mathsf T}JB+C^{\mathsf T}JD
 &=M^{-\mathsf T}(R^{\mathsf T}H_B+H_D)=0,\\
 A^{\mathsf T}JA+C^{\mathsf T}JC
 &=M^{-\mathsf T}(R^{\mathsf T}H_BR+H_D)M^{-1}\\
 &=M^{-\mathsf T}H_M M^{-1}=J.
\end{align*}
Therefore \(S\in\Sp_{\rm reg}(V\oplus V)\).
Finally, its cross-ratio is the prescribed matrix:
\[
 B^{-1}AC^{-1}D=RM^{-1}(MD^{-1})D=R.
\]
This proves surjectivity and completes the orbit-space bijection.
\end{proof}

The relation to the classical double-coset
classification~\cite[Proposition~3.1]{GoldsteinGuralnick2007} is
$B\mapsto PBT^{-1}$, with $H=B^{\mathsf T}JB$ and $J-H$ nonsingular.
Conversely, choose $D$ with $D^{\mathsf T}JD=J-H$.  The column
$\binom{B}{D}$ is then a symplectic embedding.  A symplectic basis of
its orthogonal complement supplies $\binom{A}{C}$.  If $Cx=0$,
orthogonality gives $B^{\mathsf T}JAx=0$, hence $Ax=0$ and $x=0$;
the same argument with $D$ proves that $A$ is invertible.  Thus this
completion is regular.  At fixed $B$, two choices of $D$ differ by a
symplectic operation at the second output, and two complement bases
differ by one at the first input.  The four-port orbits are therefore
the admissible double cosets of $B$, with the change of invariant
\[
 J^{-1}H=(I-\mathscr R)^{-1}.
\]
The proof above gives the corresponding port operations explicitly;
these operations give the channel relation in
Eq.~\eqref{eq:family-equivalence}.
For one qudit per port, the parameter \(\delta\) of the single-site
classification in Ref.~\cite{Pahari2026} corresponds to
\(\mathscr R=(\delta-1)I/\delta\).

\subsection{Universal transport and convolution}
\label{app:universal-proof}

\begin{proof}[Proof of Corollary~\ref{cor:universal-transport}]
Put \(\mathscr R=\mathscr R(S)\) and \(M=C^{-1}D\).
We prove \(1\Longleftrightarrow2\), then
\(2\Longleftrightarrow3\), and finally construct a convolution
representative for each scalar value of the cross-ratio.

\emph{1. Universal transport and a scalar cross-ratio.}
By Theorem~\ref{thm:centralizer-transport}, condition~1 is equivalent
to \(\mathscr R\) commuting with every element of \(\Sp(V)\).
We first show that an endomorphism \(N\) with this property must be
scalar.  For any nonzero \(v\in V\), consider the transvection
\[
 T_v(x)=x+[v,x]v.
\]
Since \([v,v]=0\), its inverse is
\(T_v^{-1}(x)=x-[v,x]v\).  Moreover,
\begin{align*}
 [T_vx,T_vy]=[x,y]+[v,y][x,v]+[v,x][v,y]=[x,y],
\end{align*}
where the last equality uses \([x,v]=-[v,x]\).
Thus \(T_v\in\Sp(V)\), so \(NT_v=T_vN\).  Expanding this
identity and cancelling \(Nx\) gives
\begin{equation*}
 [v,x]Nv=[v,Nx]v.
\end{equation*}
For each nonzero \(v\), nondegeneracy supplies \(x\) with
\([v,x]=1\).  Hence \(Nv=\lambda_v v\) for some \(\lambda_v\in\F\).
If \(u,v\) are linearly independent, linearity gives
\[
 \lambda_{u+v}(u+v)
 =N(u+v)=\lambda_u u+\lambda_v v.
\]
Thus \(\lambda_u=\lambda_v=\lambda_{u+v}\).  Proportional nonzero
vectors also have the same eigenvalue, so \(N\) is scalar.

Applying this conclusion to \(N=\mathscr R\) shows that condition~1
implies \(\mathscr R=rI\).  Conversely, a scalar \(\mathscr R\)
commutes with every \(F\in\Sp(V)\), so Theorem~\ref{thm:centralizer-transport}
gives condition~1.  Lemma~\ref{lem:regular-block-identities} makes
both \(\mathscr R\) and \(I-\mathscr R\) invertible, excluding
\(r=0,1\).  This proves the equivalence of conditions~1 and~2.

\emph{2. The conformal symplectic condition.}
To obtain condition~3 from condition~2, substitute \(\mathscr R=rI\)
into Eqs.~\eqref{eq:HB-R}--\eqref{eq:HL-R}:
\begin{align*}
 B^{\mathsf T}JB&=(1-r)^{-1}J,\\
 D^{\mathsf T}JD&=-r(1-r)^{-1}J,\\
 (C^{-1}D)^{\mathsf T}J(C^{-1}D)&=-rJ.
\end{align*}
All three multipliers are nonzero because \(r\ne0,1\).  By the
definition of \(\CSp(V)\), these identities give condition~3.
For BGJ representatives, the corresponding similitude identities
for the characteristic-function maps are recorded in
Ref.~\cite[Appendix~C]{XiongKimZhangFeiWu2026}.
Conversely, condition~3 gives, in particular,
\(M^{\mathsf T}JM=\mu(M)J\), with \(\mu(M)\ne0\).
Equation~\eqref{eq:HL-R} therefore yields
\[
 -J\mathscr R=\mu(M)J,
~ \mathscr R=-\mu(M)I.
\]
Regularity excludes \(r=0,1\), proving condition~2.

\emph{3. A convolution representative.}
It remains to realize each scalar cross-ratio by a convolution interaction.
For \(r\in\F\setminus\{0,1\}\), choose
\[
 G_r=\begin{pmatrix}r&1\\1&1\end{pmatrix}.
\]
All four entries are nonzero and \(\det G_r=r-1\ne0\), so \(G_r\)
is positive and \(S_{G_r}\) is regular.  Multiplying the blocks in
Eq.~\eqref{eq:convolution-blocks} gives the ratio
\(g_{00}g_{11}/(g_{01}g_{10})=r\) on both momentum and position
coordinates.  Consequently,
\[
 \mathscr R(S)=rI=\mathscr R(S_{G_r}).
\]
Theorem~\ref{thm:local-orbit-classification}, with \(T=I_V\), gives
the claimed ordered local equivalence.
\end{proof}

The positive stabilizer convolutions of
Ref.~\cite[Lemmas~3.1--3.2 and Theorem~3.5]{SunJinJin2026} have
\[
 F_t=\begin{pmatrix}\alpha&t\alpha\\-t\beta&\beta\end{pmatrix},
~ (1+t^2)\alpha\beta=1,
\]
with \(t,\alpha,\beta\ne0\).  Hence
\(r(F_t)=-t^{-2}\), and \(r(F_t)=r(F_{t'})\) exactly when
\(t'=\pm t\).  Their sign-orbit identification is therefore the
scalar cross-ratio identification on this family.  The
fixed-environment channel equivalences preserve optimized information
at every block length.

\section{Pure stabilizer environments}
\label{app:pure-stabilizer}

\subsection{Channel normal form and capacities}
\label{app:stabilizer-normal-proof}

\begin{proof}[Proof of Theorem~\ref{thm:stabilizer-regular}]
Fix a regular $S$ and a pure stabilizer environment $\psi_L$, and
write $\cN=\cN_{S,\psi_L}$.
We first identify the output Weyl labels and compute their restricted
symplectic rank.  An input Clifford then removes the change of labels,
and a symplectic basis exhibits the three channel factors.  Finally,
we derive their capacities without assuming capacity additivity.

\emph{1. Stabilizer character and output support.}
Write $\psi_L=\Ad_{w(u)}(\psi_L^0)$, where
$\Xi_{\psi_L^0}=\mathbf1_L$, and set
$\cN_0=\cN_{S,\psi_L^0}$.
The Weyl covariance in Eq.~\eqref{eq:weyl-covariance} gives
\[
 \cN=\Ad_{w(Bu)}\circ\cN_0.
\]
Thus the character changes only an output displacement.
Set $K=(B^\sharp)^{-1}L$.  The characteristic-function identity
\eqref{eq:characteristic-product} gives
\begin{align*}
 \Xi_{\cN_0(\rho)}(z)=\Xi_\rho(A^\sharp z)\mathbf1_L(B^\sharp z)=\mathbf1_K(z)\Xi_\rho(A^\sharp z).
\end{align*}
Regularity implies $\dim K=\dim L=n$.  The retained labels form a
linear subspace $K$, which need not be isotropic.

\emph{2. Restricted rank and the cross-ratio.}
Put $\mathscr R=\mathscr R(S)$ and $T_B=(B^\sharp)^{-1}$.
The columns of $T_BE$ form a basis of $K$.
Equation~\eqref{eq:HB-R} yields
\begin{align*}
 T_B^\sharp T_B
 &=B^{-1}(B^\sharp)^{-1}=(B^\sharp B)^{-1}=I-\mathscr R,\\
 T_B^{\mathsf T}JT_B&=J(I-\mathscr R),\\
 (T_BE)^{\mathsf T}J(T_BE)&=-E^{\mathsf T}J\mathscr R E,
\end{align*}
where the last equality uses $E^{\mathsf T}JE=0$.
Consequently $2q=\operatorname{rank}(E^{\mathsf T}J\mathscr R E)$
is the rank of the alternating form restricted to $K$; it is even
and independent of the basis $E$.
The Lagrangian identity $L^\perp=L$ also identifies its kernel:
for $c\in\F^n$,
\begin{align*}
 E^{\mathsf T}J\mathscr R Ec=0
 &\iff [v,\mathscr R Ec]=0\quad\text{for all }v\in L\\
 &\iff \mathscr R Ec\in L^\perp=L\\
 &\iff Ec\in L\cap\mathscr R^{-1}L.
\end{align*}
Since $E:\F^n\to L$ is an isomorphism, rank--nullity gives
\[
 2q=n-\dim(L\cap\mathscr R^{-1}L).
\]
In particular, $0\le q\le\lfloor n/2\rfloor$.
The radical $R_0=K\cap K^\perp$ has dimension $a=n-2q$.

\emph{3. Symplectic extension and channel normal form.}
We must remove the input-label map $A^\sharp|_K$ by extending it to
a symplectic map on $V$.
For $z,z'\in K$, symplecticity of $S^{-1}$ and isotropy of
$B^\sharp K=L$ give
\begin{align*}
 [z,z']
 &=[S^{-1}(z,0),S^{-1}(z',0)]_{V\oplus V}\\
 &=[A^\sharp z,A^\sharp z']+[B^\sharp z,B^\sharp z']\\
 &=[A^\sharp z,A^\sharp z'].
\end{align*}
Regularity makes $\phi=A^\sharp|_K$ injective; the displayed
identity shows that it preserves the restricted form.  The following basis construction gives its extension,
including when $R_0\ne\{0\}$.

Choose a linear complement $H$ with $K=H\oplus R_0$.
If $h\in H$ is orthogonal to $H$, it is orthogonal to all of $K$,
and hence belongs to $H\cap R_0=\{0\}$.
Therefore $H$ is nondegenerate and has dimension $2q$.
Choose a symplectic basis $h_i,k_i$ of $H$, with
$[h_i,k_j]=\delta_{ij}$ and the other pairings zero, and a basis
$r_1,\ldots,r_a$ of $R_0$.
Nondegeneracy of $H^\perp$ gives $t_j\in H^\perp$ satisfying
$[r_i,t_j]=\delta_{ij}$: the independent linear functionals
$[r_i,\cdot]$ define a surjection $H^\perp\to\F^a$.
To make the $t_j$ mutually orthogonal, put
\[
 b_{ij}=[t_i,t_j],~
 s_i=t_i-\frac12\sum_jb_{ij}r_j.
\]
Then
\[
 [r_i,s_j]=\delta_{ij},~
 [s_i,s_j]=b_{ij}-\tfrac12b_{ij}+\tfrac12b_{ji}=0.
\]
The $q+a$ pairs just constructed span a nondegenerate subspace;
its symplectic orthogonal complement has dimension $2q$ and supplies
the remaining $q$ pairs.  Order this symplectic basis with the $H$
pairs first, the $(r_j,s_j)$ pairs next, and the remaining pairs last,
and form $\Gamma$ with all first members before all second members.
In standard coordinates it gives
\[
 \Gamma^{\mathsf T}J\Gamma=J,~\Gamma K_{\rm std}=K,
\]
where
\[
 K_{\rm std}
 =\operatorname{Span}\{e_i,f_i:1\le i\le q\}
 \oplus\operatorname{Span}\{e_i:q<i\le n-q\}.
\]
Here $e_i,f_i$ denote the standard coordinate pairs in $V$.
Starting with $\phi h_i,\phi k_i$ on $H$ and $\phi r_j$ on $R_0$,
repeat the completion in the image space to obtain a basis matrix
$\Gamma'$.  The corresponding columns on $K$ agree with $\phi$, so
\[
 F=\Gamma'\Gamma^{-1},~F^{\mathsf T}JF=J,~F|_K=\phi.
\]
This $F$ acts on the signal input; no environmental transport
condition is required.

Choose an exact Clifford lift $U_F$.
Since $\Xi_{U_F\rho U_F^\dagger}(y)=\Xi_\rho(F^{-1}y)$,
\begin{align*}
 \Xi_{\cN_0(U_F\rho U_F^\dagger)}(z)
 &=\mathbf1_K(z)\Xi_\rho(F^{-1}A^\sharp z)\\
 &=\mathbf1_K(z)\Xi_\rho(z).
\end{align*}
Fourier inversion gives $\cN_0\circ\Ad_{U_F}=\cE_K$, where
\[
 \cE_K(X)=d^{-n}\sum_{z\in K}\Tr[Xw(-z)]w(z).
\]
Using $\Gamma'=F\Gamma$ and the Weyl action of the exact lifts, we obtain
\[
 \Ad_{U_\Gamma^\dagger w(-Bu)}\circ\cN\circ\Ad_{U_{\Gamma'}}
 =\cE_{K_{\rm std}}.
\]
This equality specifies the input encoding and output decoding.

It remains to identify the resulting channel.  The tensor-product
identification in Eq.~\eqref{eq:sector-weyl-factorization}, applied
to the coordinate pairs, gives the usual qudit tensor factors.
On a single-qudit Weyl operator,
\begin{align*}
 \id_d(w(p,t))&=w(p,t),\\
 \Delta_Z(w(p,t))&=\delta_{t,0}w(p,t),\\
 \Omega_d(w(p,t))&=\delta_{p,0}\delta_{t,0}w(p,t).
\end{align*}
Their tensor product therefore has the same action as
\[
 \cE_{K_{\rm std}}(w(z))=\mathbf1_{K_{\rm std}}(z)w(z).
\]
Since Weyl operators form a basis of the full operator space,
\[
 \cE_{K_{\rm std}}
 =\id_d^{\otimes q}\otimes
 \Delta_Z^{\otimes a}\otimes\Omega_d^{\otimes q}.
\]
The equality holds for arbitrary inputs, including inputs entangled
across these factors, and proves Eq.~\eqref{eq:stabilizer-channel-normal}.

\emph{4. Capacities.}
The $q$ identity factors transmit $q\log_2d$ qubits per use.
To bound private capacity, we show that the dephased labels supply
no private information.  A Stinespring isometry for the normal form is
\[
 |i,x,j\rangle\longmapsto
 |i,x\rangle_{B_QB_C}\otimes|\Phi_{d^q}\rangle_{B_DE_0}
 \otimes|x,j\rangle_{E_CE_D},
\]
where $i,j$ label $q$ qudits, $x$ labels $a$ qudits, and
$|\Phi_{d^q}\rangle=d^{-q/2}\sum_{h=1}^{d^q}|h,h\rangle$.
Bob and Eve both have the dephased label $x$, while Bob's $B_D$
register is fixed at $I/d^q$.

For $r$ uses, let $Y$ index any input ensemble, allowing correlated
inputs, and let $X$ be the complete dephased label.
Write $p_x=\Pr[X=x]$, and let $\chi(Y;B)$ denote the Holevo
information of the output ensemble, with analogous notation for
conditional ensembles.  The variable register $B_Q$ has dimension
$d^{rq}$, so
\begin{align*}
 \chi(Y;B)-\chi(Y;E)
 &\le\chi(Y;B)-I(Y;X)\\
 &=\sum_xp_x\chi(Y;B_Q\mid X=x)\\
 &\le\sum_xp_x\log_2\dim B_Q=rq\log_2d.
\end{align*}
The steps use data processing for Eve's copy of $X$, the entropy
formula for a state with a classical label, and the dimension bound.
Optimization and regularization give
\[
 q\log_2d\le Q(\cN)\le P(\cN)\le q\log_2d.
\]
For classical communication, Bob's variable registers $B_QB_C$
have dimension $d^{r(n-q)}$, giving
$\chi^*(\cN^{\otimes r})\le r(n-q)\log_2d$.
Uniform computational-basis inputs on these factors attain equality,
so $C(\cN)=(n-q)\log_2d$.

\end{proof}

\subsection{Universal transport and entanglement breaking}
\label{app:universal-stabilizer-proof}

\begin{proof}[Proof of Corollary~\ref{cor:universal-stabilizer-eb}]
Fix a regular $S$ and write $\mathscr R=\mathscr R(S)$.

\emph{(a) Fixed environment.}
Theorem~\ref{thm:stabilizer-regular} gives the channel normal form
and $Q=q\log_2d$.
When $q=0$, the normal form is complete dephasing.  For any state
$\omega_{RA}$ with a reference $R$,
\[
 (\id_R\otimes\Delta_Z^{\otimes n})(\omega_{RA})
 =\sum_x\omega_R^{(x)}\otimes|x\rangle\langle x|,
\]
where $\omega_R^{(x)}=\langle x|\omega_{RA}|x\rangle$.
These operators are positive and their traces sum
to one, so the output is separable.  If $q>0$, an identity factor
preserves a maximally entangled pair with a reference, and the
channel is not EB.  This proves $\mathrm{EB}\iff q=0\iff Q=0$.
For the remaining equivalence, let $E$ be a basis matrix of $L$.
Since $L^\perp=L$ and $\mathscr R$ is invertible,
\begin{align*}
 q=0
 &\iff E^{\mathsf T}J\mathscr RE=0\\
 &\iff [x,\mathscr Ry]=0\quad(x,y\in L)\\
 &\iff \mathscr RL\subseteq L^\perp=L\\
 &\iff \mathscr RL=L.
\end{align*}

\emph{(b) All pure stabilizer environments.}
Corollary~\ref{cor:universal-transport} gives
(1)$\Leftrightarrow$(2).  A scalar $\mathscr R$ preserves every
Lagrangian, so part~(a) gives
(2)$\Rightarrow$(3)$\Rightarrow$(4).

To prove (4)$\Rightarrow$(2), assume that every pure stabilizer
environment induces a channel with zero quantum capacity.
Every Lagrangian $L\subset V$ admits a pure stabilizer state.
Part~(a) therefore gives
$\mathscr R L=L$ for every Lagrangian $L$.
We now show that a linear map with this property must be scalar.

Fix $x\ne0$ and let $y\in x^\perp$.
The isotropic subspace $\operatorname{Span}\{x,y\}$ extends to
a Lagrangian $L$.  Since $\mathscr R$ preserves $L$, both
$\mathscr R x$ and $y$ belong to $L$, and hence
$[y,\mathscr R x]=0$.
This holds for every $y\in x^\perp$, so nondegeneracy of the
symplectic form implies
\[
 \mathscr R x\in(x^\perp)^\perp=\operatorname{Span}\{x\}.
\]
Thus $\mathscr R x=r_xx$ for some $r_x\in\F$.
For independent nonzero $x,y$, linearity gives
\[
 r_{x+y}(x+y)=\mathscr R(x+y)=r_xx+r_yy,
\]
and independence forces $r_{x+y}=r_x=r_y$.
For dependent nonzero vectors, the same equality of scalars follows
by scaling.  Consequently $\mathscr R=rI_V$ for a single $r\in\F$.
Regularity makes both $\mathscr R$ and $I-\mathscr R$ invertible
by Lemma~\ref{lem:regular-block-identities}, so $r\notin\{0,1\}$.
This proves (4)$\Rightarrow$(2) and completes the equivalence.
\end{proof}

\subsection{Entanglement between two cross-ratio sectors}
\label{app:two-sector-entanglement}

\begin{proof}[Proof of Corollary~\ref{cor:two-sector-entanglement}]
We first obtain the reduced states in the sector coordinates, then
compute their entropies, and finally identify the same integer in
the channel-capacity formula.

\emph{Reduced states.}
Let $L$ be the stabilizer subspace of $\psi$, let
$\chi(v)=\Xi_\psi(v)$ for $v\in L$, and set
$L_i=L\cap V_i$ and $a_i=\dim L_i$.
Recall that $T_i:\F^{2n_i}\to V_i$ are symplectic isomorphisms,
where $\dim V_i=2n_i$ and $n_1+n_2=n$, and that
$T(z_1,z_2)=T_1z_1+T_2z_2$.
Write $w_i$ for the Weyl operators on the corresponding subsystem
$B_i$.  The defining relation for $U_T$ is
\[
 U_T^\dagger w(T_1z_1+T_2z_2)U_T
 =w_1(z_1)\otimes w_2(z_2).
\]
It identifies the symplectic sectors with tensor factors of the
environment.  Thus $\psi_{B_i}$ is the $i$th reduction of the pure
state $\widetilde\psi=U_T^\dagger\psi U_T$.  Substitute this relation into
$\psi=d^{-n}\sum_{v\in L}\chi(v)w(v)$ and use
$\Tr w_i(z)=d^{n_i}\delta_{z,0}$.  The partial trace removes every
term with a nonzero Weyl label on the discarded factor, leaving
\[
 \psi_{B_i}=d^{-n_i}\sum_{v\in L_i}
       \chi(v)w_i(T_i^{-1}v).
\]
Thus only stabilizers contained entirely in $V_i$ constrain the
$i$th reduced state.

\emph{Entanglement entropy.}
Define
\[
 \Pi_i=d^{-a_i}\sum_{v\in L_i}\chi(v)w_i(T_i^{-1}v).
\]
Since $L_i$ is isotropic, its Weyl operators multiply without phase;
the character law for $\chi$ therefore gives
\begin{align*}
 \Pi_i^2&=d^{-2a_i}\sum_{s\in L_i}
 d^{a_i}\chi(s)w_i(T_i^{-1}s)=\Pi_i.
\end{align*}
Here each $s\in L_i$ occurs as $u+v$ for $d^{a_i}$ ordered pairs
$(u,v)\in L_i^2$.
Replacing $v$ by $-v$ gives $\Pi_i^\dagger=\Pi_i$, and only the
identity Weyl operator contributes to its trace.  Hence
\[
 \operatorname{rank}\Pi_i=\Tr\Pi_i=d^{n_i-a_i},~
 \psi_{B_i}=\frac{\Pi_i}{d^{n_i-a_i}}.
\]
The nonzero spectrum is uniform, so
$S(\psi_{B_i})=(n_i-a_i)\log_2d$.
Purity of $\widetilde\psi$ makes the two reduced spectra equal by
Schmidt decomposition.  Consequently
\[
 k:=n_1-a_1=n_2-a_2,~
 S(\psi_{B_1})=S(\psi_{B_2})=k\log_2d.
\]

\emph{Channel capacity.}
Write $\mathscr R=\mathscr R(S)$.
It remains to prove that the integer $q$ in
Theorem~\ref{thm:stabilizer-regular} equals $k$.
If $x=x_1+x_2\in L\cap\mathscr R^{-1}L$, then both $x$ and
$\mathscr Rx$ belong to $L$.  The distinct eigenvalues allow us to
extract their sector components:
\[
 x_1=\frac{\mathscr Rx-r_2x}{r_1-r_2}\in L_1,
~
 x_2=\frac{r_1x-\mathscr Rx}{r_1-r_2}\in L_2.
\]
Conversely, $x_i\in L_i$ implies
$x_1+x_2,\,r_1x_1+r_2x_2\in L$.  Therefore
\[
 L\cap\mathscr R^{-1}L=L_1\oplus L_2.
\]
The geometric rank formula in
Theorem~\ref{thm:stabilizer-regular} now gives
\[
 2q=n-a_1-a_2=(n_1-a_1)+(n_2-a_2)=2k.
\]
Its capacity formula yields $Q=P=k\log_2d$, as required.
\end{proof}

\subsection{The two-qudit example}
\label{app:heterogeneous}

We verify Example~\ref{ex:two-sector-code} by constructing input and output Clifford
coordinates in which the isometry transmits one input coordinate
to each output and prepares a fixed Bell pair on the other two.
For the couplings in Eq.~\eqref{eq:inhomogeneous-example} over
$\mathbb F_5$, the position mixers are
\[
 M_1=\begin{pmatrix}2&-1\\-1&1\end{pmatrix},
 ~ M_2=\begin{pmatrix}4&2\\2&3\end{pmatrix}.
\]
For a basis input $|i_1,i_2\rangle$ and one summand $|j,j\rangle$ of
the Bell environment, the output coordinates are therefore
\begin{align*}
 x=2i_1-j,~y=4i_2+2j,~
 b=-i_1+j,~c=2i_2+3j.
\end{align*}
The combinations $4x+2y$ and $b+3c$ cancel the environmental
coordinate $j$.  They motivate the input coordinates
\[
 u=3i_1+3i_2,~v=i_2-i_1,
\]
whose inverse is $i_1=u+2v$, $i_2=u+3v$.  In particular, the code
$|i\rangle\mapsto|i,i\rangle$ in the main text is $|u,0\rangle$
with $u=i$.

To make these two transmitted coordinates explicit and match the
remaining coordinates, apply the output basis permutations
\begin{align*}
 (x,y)_{A'}&\longmapsto(4x+2y,\ 3x+2y),\\
 (b,c)_{B'}&\longmapsto(b+3c,\ 4b+4c).
\end{align*}
The input and output matrices have determinants $1,2,2$, respectively.
Each induces a Clifford basis permutation, whose Weyl action is
$(p,q)\mapsto(M^{-\mathsf T}p,Mq)$.

We now check the proposed coordinates directly:
\begin{align*}
 4x+2y&=3i_1+3i_2=u,\\
 3x+2y&=i_1+3i_2+j,\\
 b+3c&=-i_1+i_2=v,\\
 4b+4c&=i_1+3i_2+j.
\end{align*}
The second coordinate is the same at both outputs.  For fixed
$(i_1,i_2)$, the substitution $z=i_1+3i_2+j$ is a bijection of
$\mathbb F_5$, so reindexing the uniform Bell sum gives the isometry
\begin{equation*}
 |u,v\rangle\longmapsto
 |u\rangle_{A'_1}|v\rangle_{B'_1}
 \otimes|\Phi_5\rangle_{A'_2B'_2}.
\end{equation*}
Denote the receiver channel in these input and output coordinates
by $\widetilde{\cN}$.  Tracing $B'$ gives
\[
 \widetilde{\cN}(X)=\Tr_2(X)\otimes I/5
       =(\id_5\otimes\Omega_5)(X).
\]
The capacities follow from the proof of
Theorem~\ref{thm:stabilizer-regular} in
Appendix~\ref{app:stabilizer-normal-proof} with $n=2,q=1$.

\section{Proof of the environmental normal form}
\label{app:normal-form}

For a state $\rho$ with $\dim S_\rho=\ell$, we first turn its
unit-modulus Weyl expectations into exact support constraints.
Clifford coordinates then separate $\ell$ pure stabilizer factors.
Finally, we identify the mean state and prove that it minimizes
relative entropy over MSPS.

\emph{Support and stabilizer character.}
The first step is to show that every unit-modulus expectation fixes
a Weyl eigenvalue on the entire support of $\rho$.
For a unitary $U$, a state $\rho$, and $c=\Tr(\rho U)$ with $|c|=1$,
\[
 \|(U-cI)\rho^{1/2}\|_2^2
 =2-2\operatorname{Re}(\overline c\Tr(\rho U))=0,
\]
where $\|X\|_2^2=\Tr(X^\dagger X)$.
Thus $U=cI$ on $\operatorname{supp}\rho$; the converse is immediate.
For $z\in S_\rho$, put $\eta(z)=\Tr[\rho w(z)]$.
Applying Weyl commutation to a nonzero vector in the support gives
\[
 \eta(z)\eta(z')=\zeta^{[z,z']}\eta(z')\eta(z)
 \quad\Longrightarrow\quad [z,z']=0.
\]
The identities $w(z)w(z')=w(z+z')$ for commuting directions and
$w(az)=w(z)^a$ then show that $S_\rho$ is an isotropic subspace and
$\eta$ is a character.  Write $\ell=\dim S_\rho\le n$.

\emph{Canonical coordinates.}
Since $S_\rho$ is isotropic, we may put all its support constraints
on distinct computational-basis coordinates.
Extend a basis of $S_\rho$ to a symplectic basis and choose a Clifford
mapping it to the first $\ell$ $Z$-directions
\cite{HostensDehaeneDeMoor2005,Gross2006}.
If the corresponding eigenvalues are $\zeta^{a_i}$, apply $X(-a)$
to these coordinates.  The resulting Clifford $K$ satisfies
\begin{align*}
 \operatorname{supp}(K\rho K^\dagger)
 &\subseteq|0\rangle^{\otimes\ell}\otimes\mathcal H_m,~ m=n-\ell,\\
 K\rho K^\dagger&=|0\rangle\!\langle0|^{\otimes\ell}\otimes\tau.
\end{align*}
A nonzero element of $S_\tau$ would supply an additional stabilizer
direction of $\rho$, contradicting $\dim S_\rho=\ell$.
Hence $S_\tau=\{0\}$.

\emph{Mean state and relative-entropy minimum.}
We now identify the maximally mixed state on the joint eigenspace
and compare it with every admissible MSPS.
Character orthogonality gives the joint-eigenspace projector
\[
 P_\rho=\frac1{d^\ell}\sum_{z\in S_\rho}\overline{\eta(z)}w(z),
 ~\Tr P_\rho=d^{n-\ell}.
\]
Its normalization has the characteristic function
of Eq.~\eqref{eq:mean-state}; hence
\[
 \cM(\rho)=P_\rho/d^{n-\ell},~
 K\cM(\rho)K^\dagger
 =|0\rangle\!\langle0|^{\otimes\ell}\otimes I/d^m.
\]
An MSPS can have finite relative entropy from $\rho$ only if its
support contains that of $\rho$.  For
$\omega=P_W/d^{n-k}$ with stabilizer subspace $W$, $\dim W=k$,
and specified character, this condition is
\begin{align*}
 D(\rho\Vert\omega)<\infty
 &\iff\operatorname{supp}\rho\subseteq\operatorname{ran}P_W\\
 &\iff W\subseteq S_\rho\text{ with matching character}.
\end{align*}
In this case $k\le\ell$, and evaluating $\log\omega$ on its support gives
\begin{align*}
 D(\rho\Vert\omega)
 &=(n-k)\log_2d-S(\rho)\\
 &\ge(n-\ell)\log_2d-S(\rho)\\
 &=D(\rho\Vert\cM(\rho)).
\end{align*}
Equality requires $k=\ell$, so $W=S_\rho$ with the same character.
Thus $\cM(\rho)$ is the unique minimizer, and
\[
 \MRM(\rho)=m\log_2d-S(\tau),
~ S(\cM(\rho))=m\log_2d.
\]
These formulas include $m=0$, with the one-dimensional state $\tau=1$.

The same coordinates verify the convex-mixture statement in
Sec.~\ref{sec:stabilizer-notation}:
\[
 |0\rangle\!\langle0|^{\otimes\ell}\otimes I/d^m
 =\frac1{d^m}\sum_{x\in\F^m}|0,x\rangle\!\langle0,x|.
\]
Each summand is a pure stabilizer state.  Clifford conjugation back
therefore places every MSPS in their convex hull.

\end{document}